\documentclass[11pt, letterpaper]{article}
\usepackage[T1]{fontenc}
\usepackage[utf8]{inputenc}
\usepackage[top=1in, bottom=1in, left=1in, right=1in]{geometry}
\usepackage{amsmath, amsthm, amsfonts, amssymb, mathtools, url, bm}
\usepackage{algorithm, algpseudocode}
\usepackage{comment}
\usepackage{enumerate}
\usepackage[shortlabels]{enumitem}
\usepackage{datetime, xspace, authblk, color, hyperref}
\hypersetup{
	colorlinks=true,
	linkcolor=blue,
	citecolor= blue
}

    \newtheorem{theorem}{Theorem}[section]
    \newtheorem{lemma}[theorem]{Lemma}
    \newtheorem{proposition}[theorem]{Proposition}
    \newtheorem{corollary}[theorem]{Corollary}
    \newtheorem{claim}[theorem]{Claim}
    
    \newtheorem{fact}[theorem]{Fact}
    \newtheorem{example}[theorem]{Example}
    \newtheorem{open}[theorem]{Open Problem}
    \newtheorem{observation}[theorem]{Observation}    
    
    \theoremstyle{definition}
    \newtheorem{definition}[theorem]{Definition}
    \newtheorem{construction}[theorem]{Construction}
    \newtheorem{remark}[theorem]{Remark}
    \newtheorem{notation}[theorem]{Notation} 
    \newtheorem{assumption}[theorem]{Assumption} 

\newcommand{\cS}{\mathcal{S}}
\newcommand{\cU}{\mathcal{U}}
\newcommand{\cV}{\mathcal{V}}

\newcommand{\NN}{\mathbb{N}}

\newcommand{\ZZ}{\mathbb{Z}}

\newcommand{\sF}{\mathsf{F}}

\newcommand{\sK}{\mathsf{K}}
\newcommand{\sL}{\mathsf{L}}

\newcommand{\bg}{\mathbf{b}}
\newcommand{\bsigma}{\bm{\sigma}}

\newcommand{\ind}[2]{\lvert #1 : #2 \rvert}
\DeclarePairedDelimiter{\abs}{\lvert}{\rvert}
\DeclarePairedDelimiter{\norm}{\lVert}{\rVert}
\newcommand{\defn}[1]{\textbf{#1}}

\DeclareMathOperator{\supp}{supp}

\DeclareMathOperator{\Aut}{Aut}

\DeclareMathOperator{\argmax}{argmax}
\DeclareMathOperator{\argmin}{argmin}
\DeclareMathOperator{\Hom}{Hom}
\DeclareMathOperator{\aHom}{aHom}

\DeclareMathOperator{\poly}{poly}
\DeclareMathOperator{\Alt}{Alt}
\DeclareMathOperator{\Sym}{Sym}

\DeclareMathOperator{\dom}{dom}
\DeclareMathOperator{\sgn}{sgn}
\DeclareMathOperator{\agr}{agr}

\DeclareMathOperator{\actson}{\curvearrowright}
\DeclareMathOperator{\dotcup}{\dot{\cup}}
\DeclareMathOperator{\emptymultiset}{\{\{\}\}}

\DeclareMathOperator{\SubSum}{SubSum}
\DeclareMathOperator{\Conj}{Conj}
\DeclareMathOperator{\Sub}{Sub}
\DeclareMathOperator{\SubLeq}{Sub^{\leq m}}
\DeclareMathOperator{\ConjLeq}{Conj^{\leq m}}

\newcommand{\SSF}{\mathfrak{F}}
\DeclareMathOperator{\HExt}{HExt} 

\newcommand{\exthom}{\varphi}
\newcommand{\vf}{\varphi}

\newcommand{\partialto}{\rightharpoonup}

\DeclareMathOperator{\val}{\texttt{val}}
\DeclareMathOperator{\out}{\texttt{out}}
\DeclareMathOperator{\True}{\texttt{True}}

\DeclareMathOperator{\Null}{\texttt{Null}}
\DeclareMathOperator{\Error}{\texttt{Error}}
 
\newcommand{\HomExt}{\textsc{HomExt}}

\newcommand{\HomExtPerm}{\textsc{HomExtPerm}}
\newcommand{\HomExtSym}{\textsc{HomExtSym}}
\newcommand{\HomExtThreshold}{\textsc{HomExtThreshold}}

\newcommand{\MultiSS}{\textsc{MultiSSR}}
\newcommand{\OrMultiSS}{\textsc{OrMultiSSR}}
\newcommand{\TriOrMultiSS}{\textsc{TriOrMultiSSR}}
\DeclareMathOperator{\OMS}{OMS}
\newcommand{\Consolidate}{\textsc{Consolidate}}

\newcommand{\ILP}{\textsc{ILP}}
\newcommand{\MoveCoset}{\textsc{MoveCoset}}
\newcommand{\LexFirst}{\textsc{LexFirst}}
 
\DeclareMathOperator{\Trioracle}{\triangle}
\DeclareMathOperator{\minoracle}{\textsc{min}}
\DeclareMathOperator{\Foracle}{\mathfrak{F}\textsf{-oracle}}

\newcommand{\pcom}[1]{{\footnotesize \texttt{ \parbox[t]{\dimexpr\linewidth-\algorithmicindent*5}{\quad $\blacktriangleright$ #1 }}}}
\newcommand{\ignore}[1]{}    
\newcommand{\margincomment}[1]{}

\begin{document}
\title{\LARGE  Homomorphism Extension}
\author{Angela Wuu\thanks{\tt wu@math.uchicago.edu}}
\affil{University of Chicago}
\maketitle

\vspace{.2in}

\begin{abstract}
We define the \textsc{Homomorphism Extension (HomExt)} problem: given a partial map $\gamma: G \partialto H$, decide whether  or not 
there exists a homomorphism $\vf:  G\to H$ extending $\gamma$, i.e., $\vf|_{\dom \gamma} = \gamma$. This problem arose in the context of list-decoding homomorphism codes but is also of independent interest, both as a problem in computational group theory and as a new and natural problem in NP of unsettled complexity status.

We consider the case $H=S_m$ (the symmetric group of degree $m$), i.e., $\gamma : G \partialto H$ gives a group action by the subgroup generated by the domain of $\gamma$. 
We assume $G\le S_n$ is given as a permutation group by a list of generators. We characterize the equivalence classes of extensions in terms of a multi-dimensional oracle subset-sum problem.  From this we infer that for bounded $G$ the HomExt problem can be solved in polynomial time.

Our main result concerns the case $G=A_n$ (the alternating group of degree $n$) for variable $n$ under the assumption that the index of $M$ in $G$ is bounded by poly$(n)$.  We solve this case in polynomial time for all $m < 2^{n-1}/\sqrt{n}$. This is the case with direct relevance to list-decoding homomorphism codes (Babai, Black, and Wuu, arXiv 2018); it is a necessary component in one of the main algorithms of that paper. 

\end{abstract}


\section{Introduction}
\label{section:intro}

\textsc{Homomorphism Extension} asks whether a group homomorphism from a subgroup can be extended to a homomorphism from the entire group. We consider the case that the groups are represented as permutation groups. The complexity of this natural problem within NP is unresolved. 

\subsection{Connection to list-decoding homomorphism codes}
Our study is partly motivated by our recent work on local list-decoding homomorphism codes from alternating groups~\cite{homcodes}. For groups $G$ and $H$, the set
of {$G\to H$} (affine) homomorphisms can be viewed as a code.   The study of list-decoding such codes originates with the celebrated paper by Goldreich and Levin~\cite{GL89} and has more recently been championed by Madhu Sudan and his coauthors~\cite{GKS06, DGKS08, GS14}. While this body of work pertains to groups that are ``close to abelian'' (abelian, nilpotent, some classes of solvable groups), in~\cite{homcodes} we began the study of the case when the group $G$ is not solvable.  As a test case, we have studied the alternating groups and plan to study other classes of simple groups.

For homomorphism codes, the ``code distance'' corresponds to the maximum agreement $\Lambda$ between two homomorphisms. The list-decoding efforts described in Babai, Black, and Wuu \cite{homcodes} only guarantee returning $M\partialto H$ partial homomorphisms, defined on subgroups $M\le G$ of order $\abs{M} >\Lambda\abs{G}$. In the case of solvable groups (all previously studied cases fall in this category), maximum agreement sets are subgroups of smallest index\footnote{ Strictly speaking, this statement requires the ``irrelevant kernel'' to be trivial. The irrelevant kernel is the intersection of the kernels of all $G \to H$ homomorphisms, cf.~\cite[Section 4]{homcodes}. The $\{$solvable$\to$nilpotent$\}$ case appears in~\cite{G15}.}, so $G$ is the only subgroup of $G$ of order greater than $\Lambda$.  This is not the case, however, for groups in general; in particular, it fails for the alternating groups $A_n$ where a maximum agreement set can be a subgroup of index $\binom{n}{2}$ (but not smaller). To solve the list-decoding problem, we need to extend these partial homomorphisms to full homomorphisms, i.e., we need to solve the Homomorphism Extension Search Problem for subgroups $M$ of order $\abs{M} > \Lambda\abs{G}$ (and therefore, of small index). Indeed, a special case of the main result here (Theorem~\ref{thm:main}) is used, and is credited to this paper, in Babai, Black, and Wuu~\cite{homcodes} to complete the proof of one of the main results of that paper. For a more detailed explanation, see part (b) of Section~\ref{section:appendix-motivation}, especially Remark~\ref{rmk:roadblock-to-LD}.

\subsection{Definition and results}
\label{section:intro-results}
	
We define the \textsc{Homomorphism Extension} problem. Denote by $\Hom(G,H)$ the set of homomorphisms from group $G$ to group $H$. 

\begin{definition}
\label{def:homext-problem}
\textsc{Homomorphism Extension}  \\
\indent \textbf{Instance:} Groups $G$ and $H$ and a partial map $\gamma: G \partialto H$. \\
\indent \textbf{Solution:} A homomorphism $\vf \in \Hom(G,H)$ that extends $\gamma$, i.e., $\vf|_M = \gamma$. 
\end{definition}

The \textsc{Homomorphism Extension} Decision Problem ($\HomExt$) asks whether a solution exists. 

\begin{remark}
Our algorithmic results for $\HomExt$ solve the \textsc{Homomorphism Extension} Search Problem as well, which asks whether a solution exists and, if so, to find one.
\end{remark}

The problems as stated above are not fully specified. Representation choices of the groups $G$ and $H$ affect the complexity of the problem. For example, $G$ may be given as a permutation group, a black-box group, or a group given by a generator-relator presentation. 

For the rest of this paper we restrict the problem to permutation groups. 

\begin{definition}
$\HomExtPerm$ is the version of {\HomExt} where the groups are permutation groups \emph{given by a list of generators.} $\HomExtSym$ is the subcase of $\HomExtPerm$ where the codomain $H$ is a symmetric group.  
\end{definition} 
Membership in permutation groups is polynomial-time testable. Our standard reference for permutation group algorithms is~\cite{SeressPGA}. Section~\ref{section:pga} summarizes the results we need, including material not easily found in the literature.  Our standard reference for permutation group theory is~\cite{DM}. 

Partial maps are represented by listing their domain and values on the domain. 
Homomorphisms in $\Hom(G,H)$ are represented by their values on a set of generators of $G$. 

For a partial map $\gamma: G \partialto H$, we denote by $M_\gamma := \langle \dom \gamma \rangle$ the subgroup of $G$ generated by the domaim $\dom \gamma$ of $\gamma$. 
\begin{remark}
\label{rmk:promise}
Whether the input map $\gamma: G \partialto H$ extends as a homomorphism in $\Hom(M_\gamma, H)$ is a polynomial-time testable condition in permutation groups. See  Section~\ref{section:promise}. 
\end{remark}

Since extending to $M_\gamma \leq G$ is easy, this  paper is primarily concerned with extending a homomorphism from a subgroup to  a homomorphism from the whole group.
\begin{assumption}[Given partial map defines a homomorphism on subgroup]
Unless otherwise stated, in our analysis we assume without loss of generality that the input partial map $\gamma: G \partialto H$ extends to a homomorphism in $\Hom(M_\gamma, H)$. This is possible due to Remark~\ref{rmk:promise}. In this case, the homomorphism $\psi$ is represented by $\gamma$, as a partial map on generators of $M_\gamma$. We will think of $\psi$ as the input to $\HomExt$. We often drop the subscript on $M_\gamma$. 
\end{assumption}

Since a minimal set of generators of a permutation group of degree $n$ has no more than $2n$ elements~\cite{Bab_subgroupchain} and any set of generators can be reduced to a minimal set in polynomial time, we shall assume our permutation groups are always given by at most $2n$ generators.

We note that the decision problem {\HomExtPerm} is in NP.
\begin{open}
Is $\HomExtPerm$ NP-complete?
\end{open}
This paper considers the important subcase of the problem when $H=S_m$, the symmetric group of degree $m$.  A homomorphism $G\to S_m$ is called a \textbf{group action} (more specifically, a \textbf{$G$-action}) on the set $[m]=\{1,\dots,m\}$.

The {\HomExtSym} problem seems nontrivial even
for bounded $G$ (and variable $m$). 

\begin{theorem}[Bounded domain] \label{thm:bounded}
If $G$ has bounded order, then $\HomExtSym$ can be solved in polynomial time. 
\end{theorem}
The degree of the polynomial in the polynomial running time is exponential in $\log^2\abs{G}$. 
\begin{open}
Can $\HomExtSym$ be replaced by $\HomExtPerm$ in Theorem~\ref{thm:bounded}, i.e., can $H = S_m$ be replaced by $H \leq S_m$?
\end{open}

Our main result, the one used in our work on homomorphism codes, concerns variable $n$ and is stated next.

In the results below, ``polynomial time'' refers to $\poly(n,m)$ time. 
\margincomment{What about $G = S_n$.}
\begin{theorem}[Main]  \label{thm:main}
If $G=A_n$ (alternating group of degree $n$),  $\HomExtSym$ can be solved in polynomial time under the following assumptions. 
\begin{itemize}
 \item[(i)] The index of $M$ in $A_n$ is bounded by $\poly(n)$, and
 \item[(ii)] $m < 2^{n-1}/\sqrt{n}$, where $H = S_m$.
\end{itemize}
\end{theorem}

Under the assumptions above, counting the number of extensions is also polynomial-time. 
\begin{theorem}[Main, counting]
Under the assumption of Theorem~\ref{thm:main}, the number of solutions to the instance of $\HomExtSym$ can be found in polynomial time. 
\end{theorem}

Note the rather generous upper bound on $m$ in item (ii). Whether an instance of $\HomExtSym$ satisfies the conditions of Theorem~\ref{thm:main} can be verified in $\poly(n)$ time (see Section~\ref{section:promise}). 

We state a polynomial-time result for very large $m$ (Theorem~\ref{thm:large-m}, of which Theorem~\ref{thm:bounded} is a special case). 

\begin{theorem}[Large range] \label{thm:large-m}
If $G \leq S_n$ and $m > 2^{1.7^{n^2}}$, then $\HomExtSym$ can be solved in polynomial time. 
\end{theorem}

\subsection{Methods}

We prove  the results stated above by reducing $\HomExtSym$ to a polynomial-time solvable case of a multi-dimensional oracle version of Subset Sum with Repetition  (SSR). 
SSR asks to represent a target number as a non-negative \textit{integral linear combination} of given numbers, whereas the classical Subset Sum problem asks for a\textit{ 0-1 combination}. SSR is NP-complete by easy reduction from Subset Sum. 

We call the multi-dimensional version of the SSR problem $\MultiSS$. 
The reduction from homomorphism extension to $\MultiSS$ is the main technical contribution of the paper (Theorem~\ref{thm:main-reduction} below). 

The reduction is polynomial time; therefore, the complexity of our solutions to $\HomExtSym$ will be the complexity of special cases of $\MultiSS$ that arise. The principal case of $\MultiSS$ is one we call ``triangular'' ; this case can be solved in polynomial time. The difficulty is aggravated by exponentially large input to $\MultiSS$, to which we assume oracle access ($\OrMultiSS$ Problem). Implementing oracles calls will amount to solving certain problems in computational group theory, addressed in Section 8 of the Appendix.

The $\MultiSS$ problem takes as input a multiset $\sK $ in universe $\cU$ (viewed as a non-negative integral function $\sK: \cU \to \ZZ^{\geq 0}$ ) and a set $\SSF$ of multisets in $\cU$ . $\MultiSS$ asks if $\sK$ is a nonnegative integral linear combination of multisets in $\SSF$ (see Section~\ref{section:reduction}). The set $\SSF$ will be too large to be explicitly given (it will contain one member per conjugacy class of subgroups of $G$).  Instead, we contend with oracle access to the set $\SSF$. For a more formal presentation of $\MultiSS$ and $\OrMultiSS$, see Section~\ref{sec:subset-sum}.

From every instance $\psi$ of $\HomExtSym$ describing a group action, we will construct an $\OrMultiSS$ instance $\OMS_\psi$ (see Section~\ref{section:reduction}). In the next result, we describe the merits of this translation. 

Two permutation actions $\vf_1, \vf_2: G \to S_m$ are \defn{permutation equivalent} if there exists $h \in S_m$ such that $\vf_1(g) = h^{-1} \vf_2(g) h$ for all $g \in G$. 
\begin{theorem}[Translation] 
\label{thm:main-reduction}
For every instance $\psi \in \Hom(M, S_m)$, the instance $\OMS_\psi$ of $\OrMultiSS$  satisfies the following. 
\begin{enumerate}[(a)]
\item $\OMS_\psi$ can be efficiently computed from $\psi$. For what this means, see Section~\ref{section:reduction}. 
\item There exists a bijection between the set of non-empty classes  of equivalent (under permutation equivalence) extensions $\exthom : G \to S_m$ and the set of solutions to $\OMS_\psi$.
\item Given a solution to $\OMS_\psi$, a representative $\exthom$ of the equivalence class of extensions can be computed efficiently. 
\end{enumerate}
\end{theorem}

Here, ``efficiently'' means in $\poly(n,m)$-time. The universe $\cU$ of $\OMS_\psi$ will be the conjugacy classes of subgroups of $M$. The set $\SSF$ will be indexed by the conjugacy classes of subgroups of $G$. These sets can be exponentially large. For $G= S_n$, $\abs{\SSF} = \exp(\widetilde{\Theta}(n^2))$ by \cite{PyberSubgroupsSn}.\\

Now, it suffices to efficiently find solutions to instances $\OMS_\psi$ of $\OrMultiSS$ arising under this reduction. 

Theorem~\ref{thm:large-m} (large $m$) follows from Theorem~\ref{thm:main-reduction} and a result of Lenstra \cite{Lenstra1983} (cf. Kannan~\cite{KannanIP}), that shows $\textsc{Integer Linear Programming}$ is fixed-parameter tractable. As $\MultiSS$ can naturally be formulated as an $\abs{\cU} \times \abs{\SSF}$ integer linear program, we conclude polynomial-time solvability due to the assumed magnitude of $m$ (see Appendix, Section 7).

For Theorem~\ref{thm:main}, we will show that $\OMS_\psi$ instances satisfy the conditions of $\TriOrMultiSS$, a ``triangular'' version of $\OrMultiSS$ (see Section~\ref{section:uniqueness}). 
\begin{theorem}[Reduction to $\TriOrMultiSS$]
\label{thm:main-methods-trireduction}
If an instance $\psi$ of $\HomExtSym$ satisfies the conditions of Theorem~\ref{thm:main}, the instance $\OMS_\psi$ of $\OrMultiSS$ is also an instance of $\TriOrMultiSS$. The oracle queries can be answered in polynomial time. 
\end{theorem}
Despite only being given oracle access, $\TriOrMultiSS$ turns out to be polynomial-time solvable (see Section~\ref{section:triormultiss}, or the Appendix, Section 5). 
\begin{proposition}
\label{prop:main-methods-trisearch}
$\TriOrMultiSS$ can be solved in polynomial time. 
\end{proposition}

\begin{proposition}
\label{prop:main-methods-triunique}
If a solution to $\TriOrMultiSS$ exists, then it is unique. 
\end{proposition}

Polynomial time for an $\OrMultiSS$ problem means polynomial in  the length of $\sK$ and the length of the representation of elements of $\SSF$. For details on representating multisets, see Section~\ref{section:prelim-multisets}.

\subsection{Efficient enumeration}

The methods discussed give a more general result than claimed. Instead of solving the Search Problem, we can in fact efficiently solve the Threshold-$k$ Enumeration Problem for $\HomExtSym$. This problem asks to find the set of extensions, unless there are more than $k$, in which case output $k$ of them. 

This question is also motivated by the list-decoding problem; specifically, Threshold-2 Enumeration can be used to prune the output list. See Section~\ref{section:appendix-motivation} for details. 
We remark that solving Threshold-2 Enumeration already requires all relevant ideas in solving Threshold-$k$ Enumeration.

\begin{definition}[Threshold-$k$]
\label{def:threshold-k-problem}
For a set $\cS$ and an integer $k\ge 0$, the
\textbf{Threshold-$k$ Enumeration Problem} asks
to return the following pair $(\val,\out)$ of outputs.\\
\indent \indent If $\abs{\cS} \le k$ , return $\val=\abs{\cS}$
and $\out =\cS$\\
\indent \indent Else,
return $\val =$ ``more'' and $\out =$
a list of $k$ distinct elements of $\cS$.
\end{definition} 

Note that the Threshold-$0$ Enumeration Problem
is simply the \textbf{decision problem} ``is $\cS$ non-empty?''
while the Threshold-$1$ Enumeration Problem
includes the \textbf{search problem} (if not empty, find an
element of $\cS$).

We say that an algorithm \textbf{efficiently} solves the
Threshold-$k$ Enumeration Problem if the cost divided by $k$
is considered ``modest'' (in our case, polynomial in
the input length).

Our work on list-decoding homomorphism codes uses solutions to the \emph{Threshold-$2$ Enumeration Problem}
for the set of extensions of a given homomorphism.
With potential future applications in mind, we discuss the
Threshold-$k$ Enumeration Problem for variable~$k$.

\begin{definition}
\textsc{Homomorphism Extension Threshold-$k$ Enumeration} ({\HomExtThreshold})  is the Threshold-$k$ Enumeration Problem for the set of solutions to $\textsc{Homomorphism Extension}$ ($\HExt_G$ defined below). 
\end{definition}

\begin{notation}[$\HExt_G(\psi)$]
We will denote by $\HExt_G(\psi)$ the set of solutions to an instance $\psi$ of $\HomExt$. 
	$$\HExt_G(\psi) := \{ \exthom \in \Hom(G,H) : \exthom|_M = \psi \}.$$  
\end{notation}


The following condition strengthens the notion of efficient solutions to threshold enumeration.

\begin{definition}[Efficient enumeration]
We say that a set $\cS$ can be \defn{efficiently enumerated}
if an algorithm lists the elements of $\cS$ at modest marginal cost.
\end{definition}
The marginal cost of the $i$-th element is the time spent between producing the $(i-1)$-st and the $i$-th elements.   In this paper,  ``modest marginal cost'' will mean $\poly(n,m)$ marginal cost, where $n$ and $m$ denote the degrees of the permutation groups $G$ and $H$, respectively. 

\begin{observation}
If a set $\cS$ can be efficiently enumerated then the
Threshold Enumeration Problem can be solved efficiently.
\end{observation}
In particular, the decision and search problems
can be solved efficiently. The following theorems are the  strengthened versions of the ones stated in Section~\ref{section:intro-results}.  

\begin{theorem}[Bounded domain, enumeration] \label{thm:bounded-enum}
If $G$ has bounded order, then the set $\HExt_G(\psi)$ can be efficiently enumerated.
\end{theorem}

\begin{theorem}[Main, enumeration]  \label{thm:main-enum}
If $G=A_n$ (alternating group of degree $n$), then the set $\HExt_G(\psi)$ can be efficiently enumerated 
under the following assumptions: 
\begin{itemize}
 \item[(i)] the index of $M$ in $A_n$ is bounded by poly$(n)$, and 
 \item[(ii)] $m < 2^{n-1}/\sqrt{n}$, where $H = S_m$.
\end{itemize}
\end{theorem}

\begin{theorem}[Large range, enumeration] \label{thm:large-enum}
If $G \leq S_n$ and $m > 2^{1.7^{n^2}}$, then the $\HomExtSym$ Threshold-$k$ Enumeration Problem can be solved in $\poly(n, m, k)$) time. 
\end{theorem}

\subsection{Enumeration methods}
\label{section:intro-enum-methods}

Recall that Theorem~\ref{thm:main-reduction} gave a bijection between classes of equivalent extensions and solutions to the $\OrMultiSS$ instance. It remains to solve the Threshold-$k$ Enumeration Problem for $\OrMultiSS$, then to efficiently enumerate extensions within one equivalence class, given a representative of that class. 

\paragraph{Solutions of Threshold-$k$ for $\OrMultiSS$\\ }

Under the assumptions of Theorem~\ref{thm:main}, the instance $\OMS_\psi$ of $\OrMultiSS$ (reduced to from the $\HomExt$ instance $\psi$) will also be an instance of $\TriOrMultiSS$. Since solutions are unique if they exist (Proposition~\ref{prop:main-methods-triunique}), solving the Search Problem also solves the Threshold-$k$ Enumeration Problem for $\TriOrMultiSS$. But, the Search Problem can be solved in polynomial time by Proposition~\ref{prop:main-methods-trisearch}. 

In the case of Theorem~\ref{thm:bounded}, $\OMS_\psi$ is an integer linear program with a bounded number of variables and constraints (corresponding to classes of subgroups of $G$) and the solutions can therefore be efficiently enumerated. 

For Theorem~\ref{thm:large-enum} (thus also implying Theorem~\ref{thm:bounded-enum}), the Threshold-$k$ Enumeration Problem for the \textsc{Integer Linear Program} version of $\OMS_\psi$ can be answered in polynomial time by viewing it as an integer linear program. See Section~\ref{section:large}.

\paragraph{Efficient enumeration within one equivalence class\\}

We now wish to efficiently enumerate extensions within each class of equivalent extensions, given a representative.

Two permutation actions $\varphi_1, \varphi_2: G \rightarrow S_m$ are \defn{equivalent (permutation) actions} if there exists $\lambda \in S_m$ such that $\varphi_1(g) = \lambda^{-1} \varphi_2(g) \lambda$ for all $g \in G$. We say that two homomorphisms $\exthom_1, \exthom_2 : G \rightarrow S_m$ are \defn{equivalent extensions} of the homomorphism $\varphi: M \rightarrow S_m$ if they (1) both extend $\varphi$ and (2) are equivalent permutation actions.  




Enumerating extensions within one equivalence class reduces to the following: 
Given subgroups $K\le L\le S_m$, efficiently enumerate
coset representatives for $K$ in $L$.

This problem was solved by Blaha and Luks in the 1980s (unpublished, cf.~\cite{BlahaLuks}).  For completeness
we include the solution based on communication by Gene
Luks~\cite{luks} (see Section~\ref{section:luks}). 

We explain the connection between finding coset representatives and the classes of equivalent extensions of $\psi$. Consider an extension $\exthom_0 \in \Hom(G,S_m)$ of $\psi \in \Hom(M,S_m)$. For any $\lambda \in S_m$, the homomorphism $\exthom_\lambda $, defined as $\exthom_\lambda(g) = \lambda^{-1} \exthom(g)_0 \lambda$ for all $g \in G$, is an equivalent permutation action. First, $\exthom_\lambda = \exthom$ if and only if $\lambda \in C_{S_m}(\psi(G))$ (the centralizer in $S_m$ of the $\psi$-image of $G$, i.e., the set of elemenets of $S_m$ that commute with all elements in $\psi(G)$). The centralizer of a group in the symmetric group can be found in polynomial time (see Section~\ref{section:appendix-centralizer}). Also, $\exthom|_\lambda$ extends $\psi$ (thus is an equivalent extension to $\exthom$) if and only if $\lambda \in C_{S_m}(\psi(M))$. 

So, finding coset representatives of $K=C_{S_m}(\psi(G))$ in $L=C_{S_m}(\psi(M))$ suffices for finding all equivalent extensions. Applying the Blaha--Luks result yields the following corollary (see Section~\ref{section:within eq class}).

\begin{corollary}
Let $M \leq G \leq S_n$ and $\psi: M \rightarrow S_m$. Suppose that $\exthom_0: G \rightarrow S_m$ extends $\psi$. Then, the class of extensions equivalent to $\exthom_0$ can be efficiently enumerated.
\end{corollary}




\subsection{Acknowledgments}
I would like to thank Madhu Sudan for introducing me to the subject of list-decoding homomorphism codes. I would also like to thank Gene Luks for communicating the content of Section~\ref{section:luks}. Last but not least, I would like to thank my adviser Laci Babai for his generous support, ideas, and endless advice.

\section{Preliminaries}
\label{section:prelim}

We write $\NN$ for $\NN = \{0, 1, 2, \ldots \}$.

\subsection{Multiset notation}
\label{section:prelim-multisets}

In this paper, we will consider both sets and multisets. All sets and multisets are finite.

We typographically distinguish multisets using ``mathsf'' font, e.g., $\sF$, $\sK$ and $\sL$ denote multisets.  A multiset within a universe $\cU$ is formally a function $\sL: \cU \rightarrow \NN$. For a member $u\in \cU$ of the universe, the \defn{multiplicity} of $u$ in $\sL$ is $\sL(u)$. We say that $u$ is an \defn{element} of $\sL$ ($u \in \sL$) if $\sL(u) > 0$, i.e., if $u$ has non-zero multiplicity in $\sL$. The set of elements of $\sL$ is called the \defn{support of $\sL$, denoted by $\supp(\sL) \subseteq \cU$}. 
We algorithmically represent a multiset $\sL: \cU \rightarrow \NN$ by listing its support $\supp(\sL) \subseteq \cU$ and the values on the support, so the description is of length $\abs{\supp(\sL)} \cdot \log (\norm{\sL}_\infty) \cdot \ell$, where $\ell$ is the description length for elements of $\sL$. The \defn{size} of $\sL$ is $\norm{\sL}_1$, the 1-norm of the function $\sL: \cU \to \NN$. 

Let $\sL_1, \sL_2: \cU \to \NN$ be two multisets in the same universe. Their \defn{sum} $\sL_1 + \sL_2$  is the multiset obtained by adding the multiplicities. We say that $\sL_1$ is a \defn{submultiset} of $\sL_2$ if $\sL_1(u) \leq \sL_2(u)$ for all $u$.

Sets will continue to be denoted by standard font and defined via one set of braces $\{ \, \}$. Often it is convenient to list the elements of a multiset $\sL$ as $\{\{L_1, \ldots, L_r\} \} = \{\{ L_i : i = 1 \ldots r\} \}$ using double braces, where $L_i \in \cU$ and each $u \in \cU$ occurs $\sL(u)$ times in this list. The length $r$ of this list is the size of $\sL$. In our notation, $\{A, A\} = \{A\}$ but $\{\{A,A\}\} \neq \{\{A\}\}$. 

 A disjoint union of two sets is denoted by $\Omega = \Omega_1 \dotcup \Omega_2$.

\subsection{Group theory notation}
\label{section:prelim-grouptheory}

Let $G$ be a group. We write $M \leq G$ to express that $M$ is a subgroup; we write $N\trianglelefteq G$ to denote that $N$ is a normal subgroup. 

For $M \leq G$ and $a \in G$, we call the coset $Ma$ of $M$ a \defn{subcoset} of $G$. We define the \defn{index} of a subcoset $Ma$ in $G$ by $\ind{G}{Ma} := \ind{G}{M}$. For a subset $S$ of a group $G$, we denote by $\langle S \rangle$ the subgroup  generated by $S$.

We introduce nonstandard notation for that will be used in the rest of the paper. 
\begin{notation}[$\Sub(G)$] 
We denote the set of subgroups of $G$ by $\Sub(G) := \{ L: L \leq G \}$. 
\end{notation}

For $L \leq G$, denote by $L \backslash G : = \{ Lg : g \in G \}$ the (right) \defn{coset space} (set of right cosets). For $L, M \leq G$, denote by $L \backslash G / M := \{ L g M : g\in G\}$ the set of \defn{double cosets}. Double cosets form an uneven partition of $G$. They are important in defining the $\MultiSS$ instance from an instance of {\HomExtSym} (see Section~\ref{section:comb}).

Two subgroups $L_1, L_2 \leq G$ are \defn{conjugate in $G$} if there exists $g \in G$ such that $L_1= g^{-1} L_2 g$. The equivalence relation of conjugacy in $G$ is denoted by $L_1 \sim_G L_2$, or $L_1 \sim L_2$ if $G$ is understood. 

\begin{notation}
\label{notation:conj-class}
For a subgroup $L \leq G$,  the \defn{conjugacy class of $L$ in $G$} is denoted by $[L]_G$ (or $[L]$ if $G$ is understood), so $[L]_G := \{ L_1 \leq G :L_1 \sim_G L\}$.
\end{notation}

\begin{notation}[$\Conj(G)$] 
We denote the set of conjugacy classes of $G$ by $\Conj(G) := \{ [L] : L \leq G\}$. 
\end{notation}

Using the introduced notation, if $L\leq G$, then $L \in \Sub(G)$, $L \in [L] \in \Conj(G)$ and $[L] \subset \Sub(G)$. 
\subsection{Permutation groups} 
\label{section:prelim perm groups}
In this section we fix terminology for groups and, in particular, permutation groups. A useful structure theorem for large subgroups of the alternating groups is presented as well. For reference see \cite{DM}.

For a set $\Omega$, $\Sym(\Omega)$ denotes the symmetric group on $\Omega$ and $\Alt(\Omega)$ denotes the alternating group on $\Omega$. Often, we write $S_n$ (or $A_n$) for the symmetric (or alternating) group on $[n] = \{1, \ldots, n\}$.

\begin{definition}[Group actions]
A \defn{(permutation) action} of a group $G$ on a set $\Omega$ is given by a homomorphism $\psi: G \rightarrow \Sym(\Omega)$, often denoted by $G \overset{\psi}{\actson} \Omega$ or  $G \actson \Omega$.
\end{definition}

Let $G \leq \Sym(\Omega)$, $g \in G$, $\omega \in \Omega$, and $\Delta \subset \Omega$. 

The image of $\omega$ under $g$ is denoted by $\omega^g$. This notation extends to sets. So, $\Delta^g := \{ \omega^g: \omega \in \Delta \}$ and $\Delta^G := \{\omega^g : \omega \in \Delta, g\in G\}$. The subset $\Delta \subset \Omega$ is \defn{$G$-invariant} if $\Delta^G = \Delta$. The \defn{orbit $\omega^G$ of $\omega$ under action by $G$} is given by $\omega^G := \{ \omega^g : g \in G\}$. The orbits of $G$ are $G$-invariant and they partition $\Omega$. All $G$-invariant sets are formed by unions of orbits. 

The \defn{point stabilizer $G_\omega$ of $\omega$} is the subgroup of $G$ fixing $\omega$, given by $G_\omega = \{ g \in G \mid \omega^g = \omega \}$. The \defn{pointwise stabilizer $G_{(\Delta)}$ of $\Delta$} is the subgroup fixing every point in $\Delta$, given by $G_{(\Delta)} = \bigcap_{\omega \in \Delta} G_\omega$. The \defn{setwise stabilizer $G_\Delta$ of $\Delta$} is given by $G_\Delta =  \{ g \in G \mid \Delta^g = \Delta \}$. 

Let  $\Delta \subseteq \Omega$ be $G$-invariant. For $g \in G$, denote by $g^\Delta$ the restriction of the action of $g$ to $\Delta$. 
The group $G^\Delta = \{ g^\Delta : g \in G \} \leq \Sym(\Delta) $ is the image of the permutation representation of $G$ in its action on $\Delta$. We see that $G^\Delta \cong G/G_{(\Delta)}$.

We state a result that goes back to Jordan. Its modern formulation by Liebeck (see \cite[Theorem 5.2A]{DM}) describes the small index subgroups of $A_n$. This theorem is used to categorize group actions by $A_n$ in Theorem \ref{thm:main}. 

\begin{theorem}[Jordan--Liebeck]
	\label{thm:JordanLiebeck}
	Let $n \geq 10$ and let $r$ be an integer with $1 \leq r < n/2$. 
	Suppose that $K \leq A_n$ has index $\ind{A_n}{K} < \binom{n}{r}$. Then, for some $\Delta \subseteq [n]$ with $\lvert \Delta \rvert < r$, we have $(A_n)_{(\Delta)} \leq K \leq (A_n)_{\Delta}$. 
\end{theorem}

\subsection{Equivalent extensions} 
\label{section:prelim equivalence}
In this section we characterize equivalence of two group actions and, in particular, fix notation to describe equivalence.

\begin{definition}[Equivalent permutation actions]
Two permutation actions $G \curvearrowright \Omega$ and $G \curvearrowright \Gamma$ are \defn{equivalent} if there exists  a bijection $\zeta: \Omega \rightarrow \Gamma$ such that $\zeta(\omega^g) = (\zeta(\omega))^g$ for all $g \in G$ and $\omega \in \Omega$. 
\end{definition}
Note that two permutation actions $\psi_1, \psi_2:G \rightarrow S_m$  of $G$ on the same domain are equivalent if there exists  $\zeta \in S_m$ such that $\psi_1(g) = \zeta^{-1}\psi_2(g) \zeta$ for all $g \in G$. 

The Introduction defined two homomorphisms $\exthom_1, \exthom_2: G \rightarrow S_m$ as ``equivalent extensions'' of $\varphi:M \rightarrow S_m$ if they both extend $\varphi$ and if they are equivalent as actions. The following definition is equivalent to that definition provided in the Introduction.

For groups $M \leq G$, the \defn{centralizer} of $M$ in $G$ is given by $C_G(M) = \{ g \in G: (\forall x \in M)(g x = x g) \}$. 

\begin{definition}[Equivalent extensions] 
Let $M \leq G$ and $\psi: M \rightarrow S_m$. 
We say that $\vf_1$ and $\vf_2$ are \defn{equivalent extensions of $\vf$} if there exists $\zeta \in C_{S_m}(\psi(M))$ such that $\zeta^{-1} \vf_2(g) \zeta = \vf_1(g)$ for all $g \in G$. 
\end{definition}

Next we consider the equivalence of transitive group actions, through their point stabilizers. A $G$-action on $\Omega$ is \defn{transitive} if $\omega^G = \Omega$ for all $\omega \in \Omega$, i.e., for every pair $\omega_1, \omega_2 \in \Omega$, there is a group element $g \in G$ satisfying $\omega_1^g = \omega_2$. 
Lemma~\ref{lem_transgroupactions_equivalence} is Lemma 1.6A in \cite{DM}. 

\begin{lemma}
\label{lem_transgroupactions_equivalence}
Suppose $G$ acts transitively on the sets $\Omega$ and $\Gamma$. Let $L$ be the stabilizer of a point in the first action.  Then, the actions are equivalent if and only if $L$ is the stabilizer of some point in the second action. 
\end{lemma}

Recall that we denote the conjugacy class of a subgroup $L\leq G$ by $[L]$, so $L$ is conjugate to $L_1$ if and only if $[L] = [L_1]$. We find all point stabilizers are conjugate, and all conjugate subgroups are point stabilizers. 

\begin{fact}
Let $L$ be a point stabilizer of a transitive $G$-action on $\Omega$. A subgroup $L_1$ is conjugate to $L$ ($ [L_1]=[L]$) if and only if $L_1$ is also the stabilizer of a point in $\Omega$. 
\end{fact}

All transitive $G$-actions are equivalent to one of its natural actions on cosets, $\rho_L$ defined below. 

\begin{example}[Natural actions on cosets]
\label{ex:coset-action}
For $L \leq G$, we denote by $\rho_L$ the natural action of $G$ on $L \backslash G$. More specifically, an element $g \in G$ acts on a coset $Lh \in L \backslash G$ as $(Lh)^g := L(hg)$. 
\end{example}

We see that the equivalence class of a transitive action is determined by the conjugacy class of its point stabilizers. 
\begin{corollary}
\label{cor:prelim-equiv-vs-conj}
Consider a transitive $G$-action $\varphi:G \rightarrow \Sym(\Omega)$. Let $L \leq G$. 
The following are equivalent.
\begin{enumerate}[(1)]
\item $\varphi$ is equivalent to $\rho_L$. 
\item $L$ is a point stabilizer of the $G$-action. 
\item  Some $L_1 \leq G$ satisfying $L_1 \sim L$ is a point stabilizer of the $G$-action.  
\item $\varphi$ is equivalent to $\rho_{L_1}$ for $L_1 \sim L$. 
\end{enumerate} 
\end{corollary} 

Motivated by Corollary~\ref{cor:prelim-equiv-vs-conj}, we will define the notion of ``$(G,L)$-actions,'' which describe transitive $G$-actions up to equivalence. This definition will be generalized to intransitive actions (see Section~\ref{section:GL-actions}).   

\subsection{Computation in permutation groups}
A permutation group $G \leq S_n$ is \defn{given} by a list of generators. We say that $G$ is \defn{known} if a list of generators of $G$ is known. Based on this representation, membership testing can be performed in polynomial time. In Appendices~\ref{section:pga} and~\ref{section:luks} we list the algorithmic facts about permutation groups used in this paper. 

\section{Multi-dimensional subset sum with repetition}
\label{sec:subset-sum} 

We consider the \textsc{Subset Sum Problem with Repetitions} (SSR). An instance is given by a set of positive integers and a ``target'' positive integer $s$. The question is ``can $s$ be represented as a non-negative linear combination\footnote{Notice that a non-negative linear combination of a set of integers is exactly the sum of a multiset in that set of integers. This question is asking for the existence of a multiset.} of the other integers?'' This problem is NP-complete by an easy reduction from the standard \textsc{Subset Sum} problem, which asks instead for a 0-1 linear combination.


We define a multi-dimensional version (\MultiSS) below.  It has its own associated Decision, Search, and Threshold-$k$ Enumeration (Definition~\ref{def:threshold-k-problem}) Problems. 
\begin{definition}
\label{def:subsum-problem}
\textsc{Multi-dimensional Subset Sum  with Repetition} (\MultiSS) \\ 
\indent \textbf{Instance:} Multiset $\sK: \cU \rightarrow \NN$ and set $\SSF$ of multisets in $\cU$.\footnote{$\cU$ is the underlying universe. Its entirety is not required in the input, but its size is the dimensionality of this problem. An element  $\sF \in \SSF$ is a multiset $\sF: \cU \to \NN$ in $\cU$.} \\
\indent \textbf{Solution:} A multiset of $\SSF$ summing to $\sK$, i.e., a multiset $\sL: \SSF \rightarrow \NN$ satisfying $\sum\limits_{\sF \in \SSF} \sL(\sF) \cdot \sF = \sK$.

\begin{notation}[$\SubSum(\sK, \SSF)$]
\label{notation:SubSum}
We write $\SubSum$ for the set of solutions to an instance of $\MultiSS$, i.e., 
	$$ \SubSum(\sK, \SSF) := \left\{ \sL: \SSF \rightarrow \NN \;  \bigg\vert \; \sum\limits_{\sF \in \SSF} \sL(\sF) \cdot \sF = \sK \right\}.$$
\end{notation}

The $\MultiSS$ Decision Problem asks whether a solution exists ($\SubSum$ is nonempty). 

The $\MultiSS$ Search Problem asks whether a solution exists and, if so, find one. 

The $\MultiSS$ Threshold-$k$ Enumeration Problem asks for the solution to the Threshold-$k$ Enumeration Problem for the set $\SubSum$. 

\end{definition}

\begin{remark}[$\MultiSS$ as \textsc{Integer Program}]
Every instance of $\MultiSS$ can naturally be viewed as an instance of \textsc{Integer Linear Programming}, with $\abs{\cU}$ constraints and $\abs{\SSF}$ variables. The variables $\sL(\sF)$ are the number of copies of each $\sF \in \SSF$ in the subset sum. The constraints correspond to checking that every element in $\cU$ has the same multiplicities in $\sK$ and $\sum \sL(\sF)\cdot \sF$. 
\end{remark}

\subsection{Oracle MultiSSR} 
\label{section:ormultiss}

In our application, the set $\SSF$ and universe $\cU$ will be prohibitively large to input explicitly. To address this, we define an oracle version of $\MultiSS$  called \textsc{Oracle Multi-dimensional Subset Sum with Repetitions} (\OrMultiSS). We will reduce a $\HomExtSym$ instance $\psi$ to an {\OrMultiSS} instance denoted by $\OMS_\psi$, then show that the oracles can be answered efficiently. 

We will find it convenient to introduce a bijection between $\SSF$ and another set $\cV$ of simpler objects, used to index $\SSF$.\footnote{The index set $\cV$ will be the conjugacy classes of subgroups of $G$, whereas $\SSF$ will be a set of multisets of conjugacy classes of subgroups of $M$.} Access to $\SSF$ is given by the oracle ``$\Foracle$,''  which on input $v \in \cV$ returns the element $\sF_v$  of $\SSF$ indexed by $v$. Elements of the universes $\cU$ and $\cV$ are encoded by strings in $\Sigma_1^{n_2}$ and $\Sigma_2^{n_2}$, respectively, and the alphabets $\Sigma_i$ and encoding lengths $n_i$ constitute the input. 

We allow non-unique\footnote{In our application, $\Sigma_1 = S_n$ and $\Sigma_2 = S_m$. The universes $\cU$ and $\cV$ will be conjugacy classes of large subgroups of $S_n$ and $S_m$, respectively. Each conjugacy class is non-uniquely encoded by generators of a subgroup in the class.} encodings of $\cU$ and $\cV$, but provide ``equality'' oracles.\footnote{We will not need to test membership of a string from $\Sigma^n$ in the universe.} To handle  non-unique encodings of $\cV$ in $\Sigma_2^{n_2}$, we assume that $\Foracle$ returns the same multiset on $\cU$ (though possibly via different encodings) when handed different encodings of the same $v \in \cV$. Writing $\sK: \cU \rightarrow \NN$ implies that $\sK$ is represented as a multiset on $\Sigma_1^{n_1}$ but with the promise that all strings in its support are encodings of elements of  $\cU$. 

\begin{definition}
\textsc{Oracle Multi-dimensional Subset Sum with Repetition} ({\OrMultiSS})  \\ 
\indent \textbf{Instance:} \\
\indent \indent \underline{Explicit input} \\
\indent \indent \indent Alphabets $\Sigma_1$ and $\Sigma_2$; \\  
\indent \indent \indent Numbers $n_1, n_2 \in \NN$, in unary; and \\ 
\indent \indent \indent Multiset $\sK: \cU \rightarrow \NN$, by listing the elements in its support and their multiplicities.\\ 
\indent \indent \underline{Oracles}\\ 
\indent \indent \indent $\equiv$ oracle for equality in $\cU$ or $\cV$, and \\
\indent \indent \indent $\Foracle$ oracle for the set $\SSF = \{\sF_v: \cU \rightarrow \NN \}_{v \in \cV}$, indexed by $\cV$.\\
\indent \textbf{Solution:} A sub-multiset of $\cV$ that defines a sub-multiset of $\SSF$ summing to $\sK$, i.e., \\
\indent \indent a multiset $\sL : \cV \to \NN$ satisfying $\sum\limits_{v \in \cV} \sL(v) \cdot \sF_v = \sK $. 

\begin{notation}[$\SubSum(\sK, \SSF)$]
Again, we write $\SubSum$ for the set of solutions to an instance of $\OrMultiSS$, though the indexing is slightly different. 
	$$ \SubSum(\sK, \SSF) := \left\{ \sL: \cU \rightarrow \NN \;  \bigg\vert \; \sum\limits_{v \in \cV} \sL(v) \cdot \sF_v = \sK \right\}. $$
\end{notation}

The length of the input is $\log \abs{\Sigma_1} + \log \abs{\Sigma_2} + n_1 + n_2 + \norm{\sK}_0 \cdot \log \norm{\sK_\infty} \cdot n_1 \log \abs{\Sigma_1}$. 



\end{definition} 

Due to non-unique encodings, checking whether a multiset $\sL$ satisfies $\sum_{v \in \cV} \sL(v) \cdot \sF_v = \sK $ will actually require calling the $\equiv$ oracle, as the multisets on the left and right sides of the equation may be encoded differently.



\subsection{Triangular \MultiSS}
\label{section:triormultiss}

The Search Problem for $\OrMultiSS$ with an additional ``Triangular Condition'' (and oracles corresponding to this condition) can be solved in polynomial time. We call this problem $\TriOrMultiSS$. 
This section defines $\TriOrMultiSS$. The next section will provide an algorithm that solves the $\TriOrMultiSS$ Search Problem in polynomial time, proving Proposition~\ref{prop:main-methods-trisearch}.

Under the conditions of Theorem~\ref{thm:main} ($G= A_n$, $M\leq G$ has polynomial index, and the codomain $S_m$ has exponentially bounded permutation domain size $m < 2^{n-1}/\sqrt{n}$), a $\HomExtSym$ instance $\psi$ reduces to an instance $\OMS_\psi$ of $\OrMultiSS$ that satisfies the additional assumptions of $\TriOrMultiSS$. The additional oracles of $\TriOrMultiSS$ can be efficiently answered (see Section~\ref{section:uniqueness}).

\paragraph{Definition of $\TriOrMultiSS$\\}
The triangular condition roughly says that the matrix for the corresponding (prohibitively large) integer linear program is upper triangular.

Below we say that a relation $\preccurlyeq$ is a \defn{total preorder} if it is reflexive and transitive with no incomparable elements.\footnote{A total order also imposes antisymmetry, i.e., if $x \preccurlyeq y$ and $y \preccurlyeq x$ then $x = y$. That is the assumption we omit.}
 
\begin{definition}
\textsc{Triangular Oracle Multi-dimensional Subset Sum with Repetition} (\TriOrMultiSS) \\ 
\indent \textbf{Input, Set, Oracles, Output:} Same as $\OrMultiSS$.\\
\indent \textbf{Triangular Condition:} $\cU$ has a total preorder $\preccurlyeq$.\\ 
\indent \indent For every $v \in \cV$, the multiset $\sF_v$ contains a unique $\preccurlyeq$-minimal element $\tau(v) \in \cU$. \\ 
\indent \indent The map $\tau: \cV \to \cU$ is injective. \\ 
\indent \textbf{Additional Oracles:} \\ 
\indent \indent $\preccurlyeq$: compares two elements of $\cU$, and \\  
\indent \indent $\Trioracle : \cU \to \cV \cup \{\Error\}$ inverts $\tau$, i.e., on input $u \in \cU$ it returns 
\begin{equation}
\triangle(u) =
\begin{cases}
 \text{the unique $v \in \cV$ such that $\tau(v) = u$} & \text{ if $v$ exists} \\ 
\Error & \text{ if no such $v$ exists}. 
\end{cases}
\end{equation}
\end{definition} 

\paragraph{Integer program and uniqueness of solutions\\}

Uniqueness of solutions for $\TriOrMultiSS$ can be seen by looking at the integer linear program formulation, where variables correspond to $\cV$ and constraints correspond to $\cU$. The Triangular Condition implies that, for every variable ($v \in \cV$), there exists a unique minimal constraint ($\tau(v) \in \cU$) containing this variable.  The ordering $\preccurlyeq$ on $\cU$ gives an ordering $\preccurlyeq_\cV$ on $\cV$ by setting $v_1 \preccurlyeq_\cV v_2$ when $\tau(v_1) \preccurlyeq \tau(v_2)$. Order the variables and constraints by $\preccurlyeq_\cV$ and $\preccurlyeq$, respectively (break ties in $\preccurlyeq$ arbitrarily and have $\preccurlyeq_\cV$ respect the tie-breaking of $\preccurlyeq$). The matrix for the corresponding linear program is upper triangular. 


Hence, if the integer program has a solution, it is unique. 
It trivially follows that solving the $\TriOrMultiSS$ Search Problem also solves the corresponding Threshold-$k$ Enumeration Problem. 

\subsection{$\TriOrMultiSS$ Search Problem} 

Algorithm~\ref{alg:MultiSS} (\TriOrMultiSS) below solves the $\TriOrMultiSS$ Search Problem in polynomial time (Proposition~\ref{prop:main-methods-trisearch}). If viewing the problem as a linear program, the algorithm essentially solves the upper triangular system of equations by row reduction, except that the dimensions are too big and only oracle access is provided. 

In each iteration, {\TriOrMultiSS} finds one minimal element $u$ in $\supp(K)$. It removes the correct number $m$ of copies of $\sF_{\triangle(u)}$ from $\sK$, in order to remove all copies of $u$ from $\sK$. If this operation fails, the algorithm returns `no solution.' Meanwhile, $\sL(\triangle(u))$ is updated in each iteration to record the number of copies of $\sF_{\triangle(u)}$ removed.

 There are three reasons the operation may fail. (1) Removing all copies of $u$ from $\sK$ may not be possible through removal of $\sF_{\triangle(u)}$ (the number $m = \sK(u) / \sF_{\triangle(u)}$ of copies is not an integer). (2) $\sK$ may not contain $m$ copies of $\sF_{\triangle(u)}$ (the operation $\sK - m \cdot \sF_{\triangle(u)}$ results in negative values). (3) $\triangle(u)$ returns $\Error$ ($u$ is not in the range of $\tau$). 
 
\paragraph{Subroutines\\} 

$\minoracle(S)$: $\minoracle$ takes as input a subset $S \subset \Sigma_1^{n_1}$ and outputs one minimal element under $\preccurlyeq$. Using the $\preccurlyeq$ oracle, a $\minoracle$ call can be executed in $\poly(\abs{S})$-time.  

$\textsc{Remove}(\sK, \sF, m)$: $\textsc{Remove}$ takes as input multisets $\sF, \sK: \Sigma_1^{n_1} \to \NN$ and a nonnegative integer $m$. It returns $\sK$ after removing $m$ copies of the multiset if possible, while accounting for non-unique encodings. Otherwise, it returns `no solution.' Pseudocode for \textsc{Remove} is provided below. 

$\Consolidate(\sK_1, \ldots, \sK_n)$: $\Consolidate$ adjusts for non-unique encodings of  $\cU \rightarrow \NN$ multisets as $\Sigma_1^{n_1} \rightarrow \NN$ multisets. Given input the encoded multisets $\sK_1, \ldots, \sK_n: \Sigma_1^{n_1} \to \NN$, $\Consolidate$ outputs multisets $\widetilde{\sK}_1, \ldots, \widetilde{\sK}_n: \Sigma_1^{n_1} \to \NN$ that encode the same multisets of $\cU$, but uniquely. In other words, $\widetilde{\sK}_i$ satisfy $\widetilde{\sK}_i = \sK_i $, with their combined support $\dot{\bigcup_i}\supp( \widetilde{\sK}_i) \subset \Sigma_1^{n_1}$ containing at most one encoding per element of $\cU$.

\paragraph{Algorithm\\}
Recall that we denote the empty multiset by $\emptymultiset$. We give pseudocode for the $\textsc{Remove}$ subroutine, followed by the main algorithm. 

\vspace{.2in}
\begin{algorithmic}
		\Procedure{Remove}{$\sK, \sF, m$}
			\State $\Consolidate(\sK, \sF)$ \pcom{Remove duplicate encodings within $\supp(\sK) \cup \supp(\sF)$.}
			\State $\sK \gets \sK - m \cdot \sF$ \pcom{Execute as $\sK, \sF: \Sigma_1^{n_1} \rightarrow \ZZ$, assuming integer range}
			\If{  $\sK$ has negative values}
				\State \Return `no solution'
			\Else\,  \Return $\sK$
			\EndIf
		\EndProcedure
\end{algorithmic}

\begin{algorithm}[H]
	\caption{Triangular Oracle MultiSS}
	\label{alg:MultiSS}
	\begin{algorithmic}[1]


		\Procedure{TriOrMultiSS}{$\Sigma_1$, $n_1$, $\Sigma_2$, $n_2$, $\sK$, $\equiv$, $\preccurlyeq$, $\Foracle$, $\Trioracle$}
		\State Initialize $\sL  = \emptymultiset$ \pcom{$\sL$ is the empty multiset  of $\Sigma_2^{n_2}$}
		\State $\Consolidate(\sK)$. \label{line:MultiSS preprocess K} \pcom{Remove duplicate encodings within $\supp(\sK)$} 
		\While{$\sK \neq \emptymultiset$ }\label{line:MultiSS big while}
			\State $u \gets \minoracle(\supp(\sK))$ 
				\pcom{$u$ is a minimal element of $\sK$} 
			\If{$\triangle(u) = \Error$}
			\State \Return `no solution' 
			\Else
				
			\State $\sF \gets \Foracle_{\Trioracle(u)}$
				\pcom{$\sF$ is $\sF_v$, where $\tau(v) = u$ by Triangular Condition}
			\State $m \gets\frac{\sK(u)}{\sF(u)}$
				\pcom{$m$ is number of copies of $\sF$ to remove from $\sK$.}
			\If{ ($m \notin \NN$) or ($\textsc{Remove}(\sK, \sF, m) = $ `no solution') }
			\State \Return `no solution' 
			\Else 
			\State $\sL(\Trioracle(u)) \gets \sL(\Trioracle(u)) + m$ \label{line:MultiSS update L}
			\label{line:MultiSS update K}  
			\State $\sK \gets \textsc{Remove}(\sK, \sF, m)$ 
			\EndIf
			\EndIf
		\EndWhile
		\State \Return $\sL$
		\EndProcedure

	\end{algorithmic}
\end{algorithm}



\paragraph{Analysis\\}

The pre-processing step of Line~\ref{line:MultiSS preprocess K} can be computed in time $\abs{\supp(\sK)}^2$, by pairwise comparisons. The \textbf{while} loop of Line~\ref{line:MultiSS big while} is executed exactly $\abs{\supp(\sK)}$ number of times, for each $u \in \supp(\sK)$. 

The $\Consolidate$ call in $\TriOrMultiSS$ returns $\widetilde{\sK}:\Sigma_1^{n_1} \rightarrow \NN$, a different encoding of the multiset $\sK$ of $\cU$, such that all elements of $\supp(\widetilde{\sK})$ are uniquely encoded. This requires ${\abs{\supp(\sK)} \choose 2}$  pairwise comparisons, or, $< \abs{\supp(\sK)}^2$ calls to the $\equiv$ oracle. Similarly, the $\Consolidate$ call in $\textsc{Remove}$ can be achieved in $< \abs{\supp(\sK) \cup \supp(\sF)}^2$ calls to the $\equiv$ oracle. 




\section{Reduction  of $\HomExtSym$ to $\OrMultiSS$}
\label{section:comb} 

We define the reduction from {\HomExtSym} to $\OrMultiSS$ then prove the three parts of Theorem~\ref{thm:main-reduction}: the polynomial-time efficiency of the reduction, the bijection between classes of equivalent extensions in $\HExt(\psi)$ and the set $\SubSum(\OMS_\psi)$ of solutions to $\OMS_\psi$, and efficiency of  of defining an extension homomorphism $\exthom \in \HExt(\psi)$  from a solution $\sL \in \SubSum(\OMS_\psi)$.  \\

For notational convenience, Section~\ref{section:GL-actions} defines ``$(G, \sL)$-actions'' which describe permutation actions up to equivalence. 

Towards proving Theorem~\ref{thm:main-reduction} (a), Section~\ref{section:reduction} presents the reduction from a $\HomExtSym$ instance $\psi$ to the $\OrMultiSS$ instance $\OMS_\psi$. We define the instance $\OMS_\psi$ and show that its oracles can be answered in $\poly(n,m)$-time. 

Section~\ref{section:comb intransitive} proves the bijection claimed in Theorem~\ref{thm:main-reduction} (b), assuming the transitive case. The transitive case is proved in Sections~\ref{section:comb decomposing} and~\ref{section:comb extending}.

Section~\ref{section:def-ext} proves Theorem~\ref{thm:main-reduction} (c) by providing the algorithmic details of defining $\exthom\in \HExt(\psi)$ given a solution in $\SubSum(\OMS_\psi)$.

\subsection{$(G,\sL)$-actions, equivalence classes of $G$-actions}
\label{section:GL-actions}
We introduce the terminology ``$(G,\sL)$-actions'' (or ``$(G, L)$-actions'' for transitive actions), which describes group actions up to permutation equivalence. The $\sL: \Sub(G) \to \NN$ denotes a multiset of subgroups of $G$, describing point stabilizers of the action. We make this more precise. 

Recall that we write $[L]_G = [L]$ to denote the conjugacy class of the subgroup $L$ in $G$.
\begin{definition}[$(G,L)$-action] 
Let $\varphi:G \rightarrow \Sym(\Omega)$ be a transitive action. Let $L \leq G$. We say that $\varphi$ is a \defn{$(G,L)$-action}  if $\varphi$ is equivalent to $\rho_L$, the natural on right cosets of $L$ (Example~\ref{ex:coset-action}). We say that $\varphi$ is a \defn{$(G, [L])$-action} if $\varphi$ is a $(G,L)$-action. 
\end{definition}

By Corollary~\ref{cor:prelim-equiv-vs-conj}, a $G$-action is a $(G,L)$-action if and only if $L$ is a point stabilizer of the action. Moreoever, a $(G,L)$-action is a $(G,L_1)$-action if and only if  $[L] = [L_1]$. So, we can speak of $(G,[L])$-actions and make no distinction between $(G,[L])$-actions and $(G,L)$-actions.

We now introduce notation to describe equivalence between intransitive actions. 

\begin{definition}[$(G,\sL)$-action]
Let $\varphi: G \to \Sym(\Omega)$ be a group action. 
Let $\sL: \Sub(G) \to \NN$ be a multiset listed as $\sL =  \{\{ L_i \leq G \}\}_{i=1}^d$. 
We say the action of $G$ on $\Omega$ is a \defn{$(G, \sL)$-action} if the orbits in $\Omega$ of the action can be labeled $\Omega = \Omega_1 \dotcup \cdots \dotcup \Omega_d$ so that $G$ acts on $\Omega_i$ as a $(G,L_i)$-action for all $1 \leq i \leq d$.\footnote{The multiset $\sL:\Sub(G) \rightarrow \NN$ contains one point stabilizer per orbit of the $G$-action. Viewing $\sL$ as a multiset is essential. For example, $\sL = \{\{G\}\}$ describes the trivial action of $G$ on one point, whereas $ \sL = \{\{G,G\}\}$ describes the trivial action of $G$ on two points.} 
\end{definition}

Again, the equivalence class of the $G$-action is determined by the multiset $\sL$ up to conjugation of its elements. We introduce notation describing conjugate multisets. 

\begin{notation}
\label{notation:conj-multiset}
Let $\sL = \{ \{ L_1, \ldots, L_k\}\}$ be a multiset of subgroups of $G$. We denote by $[\sL]_G = \{ \{ [L_1]_G, \ldots, [L_k]_G\}\}$ the multiset of conjugacy classes for the subgroups of $\sL$. 
\end{notation}
In other words, for a multiset $\sL: \Sub(G) \to \NN$, denote by $[\sL]_G : \Conj(G) \to \NN$ the multiset found by replacing every element $L \in \sL$ by $[L]_G$. Multiplicities of subgroup conjugacy classes $[L]$ in the multiset $[\sL]$ satisfy $[\sL]([L]) = \sum_{L \in [L]} \sL(L)$.  We may write $[L]$ for $[L]_G$ if $G$ is understood.


\begin{definition}[Conjugate multisets]  
We say that two multisets $\sL_1, \sL_2 : \Sub(G) \rightarrow \NN$ are \defn{conjugate} if $[\sL_1] = [\sL_2]$. In other words, there exists a bijection $\pi: \sL_1 \rightarrow \sL_2$ such that $\pi(L) \sim_G L$ for all $L \in \sL_1$.\footnote{This definition does not require conjugacy of all pairs simultaneously via the one element of $G$.} 
\end{definition}
Conjugate multisets describes group actions up to equivalence, as we see in the following next statement, which follows from the definitions and Corollary~\ref{cor:prelim-equiv-vs-conj}. 

\begin{corollary}
\label{cor:similar-multiset-same-action}
Let $\sL_1, \sL_2: \Sub(G) \rightarrow \NN$. The following are equivalent. 
\begin{itemize} 
\item $\sL_1$ and $\sL_2$ are conjugate, or $[\sL_1] = [\sL_2]$. 
\item A $(G, \sL_1)$-action is permutation equivalent to a $(G, \sL_2)$-action.
\item  A $(G, \sL_1)$-action is also a $(G, \sL_2)$-action. 
\end{itemize}
\end{corollary}

So, we can speak of $(G,[\sL])$-actions and make no distinction between $(G,[\sL])$-action and $(G,\sL)$-actions.


\subsection{Reduction} 
\label{section:reduction}

In this section, we discuss the $\poly(n,m)$-time reduction from  $\HomExtPerm$ to  $\OrMultiSS$.

\begin{remark}[Meaning of ``reduction'']
As usual, our reduction will compute the explicit inputs to $\OrMultiSS$ from a $\HomExtSym$ instance in $\poly(n,m)$ time. However, to account for the oracles in $\OrMultiSS$, we provide also answers to its oracles in $\poly(n,m)$-time.
\end{remark} 


Recall that $\Sub(G)$ denotes the set of subgroups of $G$ and $\Conj(G)$ denotes the set of conjugacy classes of subgroups of $G$.  Denote by $\SubLeq(G)$ the set of subgroups of $G$ with index bounded by $m$. Denote by $\ConjLeq(G)$ the set of conjugacy classes of subgroups of $G$ with index bounded by $m$. 



\paragraph{Construction of $\OMS_\psi$}

We define $\cU, \cV, [\sK]$ and encodings $\Sigma_1^{n_1}, \Sigma_2^{n_1}$ of the $\OrMultiSS$ instance $\OMS_\psi$.\\

$\cU$: $\ConjLeq(M)$. 

$\cV$: $\ConjLeq(G)$. 

Encoding of $\cU$: words of length $n_1 = 2n$ over alphabet $\Sigma_1 = M$. A conjugacy class in $\cU$ of subgroups is encoded by a representative subgroup in $\SubLeq(M)$, which is then encoded by a list of at most $2 n$ generators.

Encoding of $\cV$: Likewise, with $\Sigma_2 = G$ and $n_2 = 2n$. 

$[\sK]$: Let $\sK : \SubLeq(M) \rightarrow \NN$ be a multiset containing one point stabilizer per orbit of the action $\psi : M \rightarrow  S_m$. So, $[\sK]: \ConjLeq(M) \rightarrow \NN$ is a multiset of conjugacy classes, as in Notation~\ref{notation:conj-multiset}. 

\paragraph{Notational issues.}  
Using $[\sK]$ versus $\sK$ reflects the non-unique encoding of $\cU = \ConjLeq(M)$ by $\SubLeq(G)$ and $\cV = \ConjLeq(G)$ by $\SubLeq(G)$, adhering to Notation~\ref{notation:conj-class} and~\ref{notation:conj-multiset}. A conjugacy class $[K] \in \cU$ will be encoded by $K \in \SubLeq(M)$. A multiset $[\sK]: \cU \to \NN$ will be encoded  by $\sK: \SubLeq(M) \to \NN$. 

 
\paragraph{Calculating $[\sK]$\\}

Calculating $[\sK]: \cU \to \NN$ from $\psi: M \to S_m$: Decompose $[m] = \Sigma_1 \dotcup \ldots \dotcup \Sigma_s$ into its $M$-orbits under the action described by $\psi$. Choose one element $x_i \in \Sigma_i$ per orbit.\footnote{The choice of $x_i$ will not affect the correctness of the reduction.} Then, calculate the multiset $\sK := \{ \{ M_{x_i} : i = 1 \ldots s \} \}$ by finding the point stabilizer of each chosen element. So, calculating $\sK$ can be accomplished in $\poly(n)$-time by Proposition~\ref{prop:pga-basic}. 

\paragraph{Answering $\equiv$ oracle\\}

The $\equiv$ oracle: given two subgroups in $\SubLeq(M)$, check their conjugacy. This can be accomplished in $\poly(n,m)$-time by Proposition~\ref{prop:pga-subgroup}.

\paragraph{Answering $\Foracle$ oracle. \\}

The set $\SSF$ is indexed by $\cV = \ConjLeq(G)$. $\Foracle$ takes as input $[L] \in \ConjLeq(G)$ (represented by a $L \in \SubLeq(G)$) and returns $[\sF_L]:\ConjLeq(M) \to \NN$ (represented by $\sF_L: \SubLeq(M) \to \NN$), defined below. The multiset $\sF_L : \SubLeq(M) \rightarrow \NN$ is defined so that $(G,L)$-actions induce $(M, \sF_L)$-actions. 

\begin{definition}[$\sF_L(\bsigma)$]
\label{def:FsubL}
Let $\bsigma = (\sigma_1, \ldots, \sigma_d)$ be a list of double coset representatives for $L \backslash G /M$. We define the multiset $\sF^M_L(\bsigma) : \Sub(M) \rightarrow \NN$ by 
	\begin{equation*}
	\sF^M_L(\bsigma) = 	\sF_L  := \{\{ \sigma_i^{-1} L \sigma_i \cap M: i = 1 \ldots d \}\}. 
	\end{equation*}
\end{definition}

In the context of extending an $M$-action $\psi:M \to S_m$ to a $G$-action, $M$ is understood, so  we drop the superscript and write $\sF_L$. 

$\Foracle$ is well-defined. First of all, the choice $\bsigma$ of double coset representatives will not affect the conjugacy class of $\sF_L^M(\sigma)$
(see Remark~\ref{rmk:Foracle-well-defined}). 
Moreover, if $[L]_G = [L_1]_G$ then $[\sF_L]_M = [\sF_{L_1}]_M$. 
Section~\ref{section:comb decomposing} further discusses and proves these claims about the properties of $\sF_L$.

$\Foracle$ can be answered in $\poly(n,m)$-time by Proposition~\ref{prop:pga-double-cosets}.




\subsection{Combinatorial condition for extensions}
\label{section:comb intransitive}
We are now equipped to state the central technical result. It relates $M$-actions to extension $G$-actions by describing how $M$-orbits may be grouped to form $G$-orbits.

First, we address the case of transitive extensions.

As in Definition~\ref{def:FsubL}, $\sF_L:\Sub(M) \to \NN$ denotes the multiset returned by the oracle $\Foracle$ on input $L \in \Sub(G)$. Since we assume the extension $G$-action is transitive, the multiset $\sF_L$ describes exactly the $M$-orbits that must be collected to form one $(G,L)$-orbit. 

\begin{lemma}[Characterization of transitive extensions]
\label{lemma:comb-trans}
Let $M, L \leq G$ and $m \in \NN$. Let $\psi : M \to S_m$ be an $M$-action.  Under these circumstances, $\psi$ extends to a $(G,L)$-action if and only if $\psi$ is a $(M,\sF_L)$-action. 
\end{lemma} 
The forward and backwards directions are Corollary~\ref{cor:comb-forward-direction} and Proposition~\ref{prop:gluing} in the next two sections. 

\begin{remark}
To rephrase Lemma~\ref{lemma:comb-trans}, an $(M, \sK)$-action extends to a transitive $(G,L)$-action if and only if $[\sK] = [\sF_L]$ (see Corollary~\ref{cor:similar-multiset-same-action}).
\end{remark}

The following result on intransitive actions is a corollary to Lemma~\ref{lemma:comb-trans}. 

\begin{theorem}[Key technical lemma: characterization of $\HomExtSym$ with codomain $S_m$]
\label{thm:HE comb main}
Let $M \leq G$ and $m \in \NN$. Let $\psi: M \to S_m$ be an $M$-action. Let $[\sL]:\Conj(G) \to \NN$. Let $[\sK] : \Conj(M) \to \NN$ describe the equivalence class of $\psi$, so $\psi$ is an $(M, \sK)$-action.  Under these circumstances, $\psi$ extends to a $(G, [\sL])$-action  if and only if $[\sK]$ is an $[\sL]$-linear combination of elements in $\SSF$, i.e., 
	\begin{equation}
	\label{eqn:HE comb char}
		[\sK] =  \sum_{L \in \sL} [\sF_L]= \sum_{[L] \in \ConjLeq(G)} \sL([L]) [\sF_L]. 
	\end{equation}

\end{theorem}


We have found that an $(M, \sK)$-action extends exactly if $\sK$ is a Subset Sum with Repetition of $\{\sK_L\}$. Compare Equation~\eqref{eqn:HE comb char} to the definition of $\SubSum(\OMS_\psi)$ (see Notation~\ref{notation:SubSum} and the reduction of Section~\ref{section:reduction}). We have found the following. 
\begin{corollary}

Let $M \leq G$ and $m \in \NN$. Let $\psi: M \to S_m$ be an $(M, [\sK])$-action, where $[\sK]: \Conj(M) \to \NN$.  Under these circumstances, $\psi$ extends to a $G$-action if and only if $\SubSum(\OMS_\psi)$ is nonempty. 
\end{corollary}

So, $\HExt(\psi)$ is nonempty if and only if $\SubSum(\OMS_\psi)$ is nonempty. 
\begin{remark}
We have found something even stronger. The multisets $[\sL]$ satisfying Equation~\eqref{eqn:HE comb char} are exactly the elements in $\SubSum(\OMS_\psi)$. A multiset $[\sL]: \Conj(G) \rightarrow \NN$ satisfies Equation~\eqref{eqn:HE comb char} if and only if $\HExt(\psi)$ contains a $(G,\sL)$-action extending $\psi$. This notation identifies all equivalent extensions, so we have found a bijection between the solutions in $\SubSum(\OMS_\psi)$ and classes of equivalent extensions in $\HExt(\psi)$, as promised by Theorem~\ref{thm:main-reduction} (b).

\end{remark}


\subsection{$(G,L)$-actions induce $(M, \sF_L)$-actions}
\label{section:comb decomposing} 

Let $M \leq G$. This section describes the $M$-action found by restricting a (transitive) $G$-action. If $\psi: G \rightarrow \Sym(\Omega)$ describes a $G$-action on $\Omega$, we will call the $M$-action on $\Omega$ found by restriction of $\psi$ to $M$ the \defn{$M$-action induced by $\psi$}, denoted by $\psi|_M$. 

First, we identify the permutation domain $\Omega$ of a $(G,L)$-action with the right cosets $L \backslash G$. By definition of ``$(G,L)$-action,'' there exists a permutation equivalence of this action with $\rho_L$ (the national action on cosets of $L$), i.e., there exists a bijection $\pi: \Omega \to L\backslash G$ respecting the $G$-action. This bijection $\pi$ identifies $\Omega$ with $L \backslash G$.    

We now describe the behavior of the induced $M$-action on $L \backslash G$. 

\begin{remark}
\label{remark:FsubL-doublecosets-vs-Morbits}
Let $M,L \leq G$. Consider the natural $M$-action on $L \backslash G$ (the $M$-action induced by the $G$-action $\rho_L$). The cosets $(Lg_1)$ and $(Lg_2)$ belong to the same $M$-orbit if and only if $Lg_1M = Lg_2 M$, i.e., if $g_1$ and $g_2$ belong to the same double coset of $L \backslash G /M$. 
\end{remark}

\begin{lemma}
\label{lemma:FsubL-Mactions}
Let $g_0 \in G$. Let $M,L \leq G$. The action of $M$ on the orbit $(Lg_0)^M$ of $Lg_0$ in $L \backslash G$ is equivalent to the action of $M$ on $K \backslash M$, where $K := g_0^{-1}Lg_0  \cap M$. The bijection is given by $La \leftrightarrow Kg_0^{-1} a$. 
\end{lemma}
\begin{proof}
Both actions are transitive. Let $\zeta: (Lg_0)^M \rightarrow K \backslash M$ be defined by $\zeta(Lg) = Kg_0^{-1}g$ for all $g \in Lg_0 M$. For all $a \in M$, 
	\begin{equation*}
	\zeta ( (Lg)^a)  = \zeta (L(ga)) = K g_0^{-1} (ga) = (Kg_0^{-1} g)^a = \zeta(Lg)^a. 
	\end{equation*}
\end{proof}

From Remark~\ref{remark:FsubL-doublecosets-vs-Morbits} and Lemma~\ref{lemma:FsubL-Mactions}, we have found the (possibly non-transitive) natural action of $M$ on $L \backslash G$ satisfies the following.
	\begin{enumerate}[(1)]
	\item The number of orbits is $\abs{L \backslash G / M}$, the number of double cosets of $L$ and $M$ in $G$. 
	\item The point stabilizer of $Lg \in L \backslash G$ under the $M$-action is $M_{Lg} = g^{-1} L g \cap M$. 
	\end{enumerate}

We restate the definition of $\sF_L$, which we now see describes the $M$-action on $L \backslash G$. 

\begin{definition}[$\sF_L(\bsigma)$]
Let $\bsigma = (\sigma_1, \ldots, \sigma_d)$ be a list of double coset representatives for $L \backslash G /M$. We define the multiset $\sF^M_L(\bsigma) : \Sub(M) \rightarrow \NN$ by 
	\begin{equation*}
	\sF^M_L(\bsigma) = 	\sF_L  := \{\{ \sigma_i^{-1} L \sigma_i \cap M: i = 1 \ldots d \}\}. 
	\end{equation*}
\end{definition}
If the subgroup $M$ is understood, we drop the superscript $M$. 

From Remark~\ref{remark:FsubL-doublecosets-vs-Morbits} and Lemma~\ref{lemma:FsubL-Mactions}, we find that $(G,L)$-actions restrict to $(M, \sF_L)$-actions. 
\begin{corollary}
\label{cor:comb-forward-direction}
Let $M, L \leq G$. Let $\bsigma = (\sigma_1, \ldots, \sigma_d)$ be a set of double coset representatives of $L \backslash G /M$. If $G$ acts on $\Omega$ as a $(G,L)$-action, then the induced action of $M$ on $\Omega$ is an $(M, \sF_L(\bsigma))$-action. In fact, the $M$-action induced by a $(G,[L])$-action is an $(M, [\sF_L])$-action. 
\end{corollary}

The last sentence of Corollary~\ref{cor:comb-forward-direction} follows from Corollary~\ref{cor:similar-multiset-same-action} and Lemma~\ref{lemma:Foracle-well-defined} below, which say that the choice $\sigma$ of double coset representatives and the choice $L$ of conjugacy class representative make no difference to the conjugacy class $[\sF_L(\bsigma)]$. \\ 

We show the $\Foracle$ is well-defined. 
\begin{remark}
\label{rmk:Foracle-well-defined}
For any two choices $\bsigma$ or $\bsigma'$ of double coset representatives of $L \backslash G /M$, we have that  $[\sF_L(\bsigma)]_M = [\sF_L(\bsigma')]_M$. So, we may reference $(M, \sF_L)$-actions without specifying $\bsigma$. 
\end{remark}

This is true since, if $\sigma_1$ and $\sigma_2$ are representatives of the same double coset, then $\sigma_1^{-1} L \sigma_1 \cap M$ and  $\sigma_2^{-1} L \sigma_2 \cap M$ are conjuate in $M$. 

In fact,  only the conjugacy class of $L$ matters in determining the conjugacy class of $\sF_L$. In particular, the $\Foracle$ oracle is well-defined. 

\begin{lemma}
\label{lemma:Foracle-well-defined}
Let $M, L, L_1 \leq G$. If $[L]_G = [L_1]_G$, then $[\sF^M_L]_M = [\sF^M_{L_1}]_M$. In other words, if $L$ and $L_1$ are conjugate in $G$, then $\sF^M_L$ and $\sF^M_{L_1}$ are conjugate in $M$. 
\end{lemma}

\begin{proof}
The natural $G$-actions on $L \backslash G$ and $L_1\backslash G$ are equivalent by Corollary~\ref{cor:prelim-equiv-vs-conj}. Thus, the induced $M$-action on $L \backslash G$ and the induced $M$-action on  $L_1\backslash G$ are equivalent, using the same bijection on the domain. But, the $M$-action on $L \backslash G$ is an $(M, \sF_L)$-action and the $M$-action on $L_1 \backslash G$ is an $(M, \sF_{L_1})$-action. 
By Corollary~\ref{cor:similar-multiset-same-action}, we find $[\sF_L]_M = [\sF_{L_1}]_M$.
\end{proof}

\subsection{Gluing $M$-orbits to find extensions to $G$-actions }
\label{section:comb extending} 
In this section we see that any $(M, \sF_L)$-action can extend to a $(G,L)$-action. 

We proved in the last section that the $M$-action induced by every $(G,L)$-action is an $(M, \sF_L)$-action. Since all $(M, \sF_L)$-actions are permutation equivalent (Corollary~\ref{cor:similar-multiset-same-action}), the given $(M,\sF_L)$-action and the $(M, \sF_L)$-action induced by the $(G,L)$-action $\rho_L$ are permutation equivalent. This gives a bijection between permutation domains which respects the $M$-actions. Thus, the given $M$-action extends to a $(G,L)$-action. 

In what follows we construct the bijection explicitly.

Let $M, L \leq G$. Let $\psi: M \to \Sym(\Omega)$ be an $(M, \sF_L)$-action. 
By definition, we may label the orbits in $\Omega$ by the sets of cosets $K \backslash M$ for $K \in \sF_L$ (each orbit is labeled by one set of cosets $K \backslash M$), so that $M$ acts as the natural action $\rho_K$ on each coset. 

Consider the natural $G$-action $\rho_L$ on right cosets $L \backslash G$. It will suffice to label $\Omega$ by the right cosets $L \backslash G$, so that the natural action of $G$ extends the $M$-action $\psi$. 
Let $\sigma \in G$. Lemma~\ref{lemma:FsubL-Mactions} gave a permutation equivalence between the $M$-action on the orbit $(L\sigma)^M$ of $(L\sigma)$ in $L \backslash G$ and the natural $M$-action on $F_i \backslash M$, where $F_i = \sigma^{-1} L \sigma \cap M$. We extend this equivalence here. 

\begin{construction}[Equivalence $\zeta$]
\label{def:extension map}
Fix a choice $\bsigma = (\sigma_1, \ldots, \sigma_d)$ of double coset representatives for $L \backslash G /M$. Recall the definition $\sF_L(\bsigma) = \{ \{ F_i : i = 1 \ldots d \} \}$, where $F_i = \sigma_i^{-1} L \sigma_i \cap M$.  Define the map $\zeta$ by 
	\begin{equation*}
	\zeta: \left(\dot{\bigcup}_i F_i\backslash M \right) \to L \backslash G, \;  \; \; \zeta: F_i \tau \mapsto L \sigma_i \tau. 
	\end{equation*}
\end{construction}

That $\zeta$ is a permutation equivalence of the $M$-actions on the two sets follows immediately from Lemma~\ref{lemma:FsubL-Mactions}. 
\begin{corollary}
The map $\zeta$ given in Construction~\ref{def:extension map} is a permutation equivalence of the $M$-action. 
\end{corollary}

The next result is almost immediate from our discussion above.

\begin{proposition}[Gluing]
\label{prop:gluing}
Let $L, M \leq G$. Suppose that $\psi: M \rightarrow \Sym(\Omega)$ describes an $(M, \sF_L)$-action. Then, there exists an extension $\vf: G \rightarrow \Sym(\Omega)$ of $\psi$ that is a $(G,L)$-action. 
\end{proposition}
\begin{proof}
We label the $M$-orbits of $\Omega$ by the cosets $F_i\backslash M$, use $\zeta$ to label $\Omega$ by $L \backslash G$, then let $G$ act on $\Omega$ in its natural action on $L \backslash G$. The output is the evaluation of $\exthom$ on the generators of $G$ as given by $\exthom(g_j): La \mapsto L a g_j$. 
\end{proof}

\subsection{Defining one extension from $\SubSum$ solution}
\label{section:def-ext} 

We prove Theorem~\ref{thm:main-reduction} (c) by defining an extension $\exthom \in \HExt(\psi)$ given a solution $[L] \in \SubSum(\OMS_\psi)$. 

First of all, Construction~\ref{def:extension map} addresses the transitive case. It gives an explicit bijection $\zeta$ that, given an $(M, \sF_L)$-action for $L \leq G$, defines an extension $(G,L)$-action. This bijection $\zeta$ can be computed in $\poly(n,m)$ time. 

The issue remains of finding the $\sF_L$ ``grouping'' of the $M$-orbits that respect the orbits of the $(G,\sL)$-action.  \\ 

Fix a $\HomExtSym$ instance $\psi$. Fix $\sL:  \Sub(G) \to \NN$ in  $\SubSum(\OMS_\psi)$, so $\sL$ satisfies Equation~\eqref{eqn:HE comb char}. 
Recall that $\sL$ is represented by listing the subgroups in its support and their multiplicities. Since $\abs{\supp(\sL)} \leq \norm{\sL}_1$, the number of orbits of the $G$-action, we find that $\abs{\supp(\sL)} \leq m$.  

It takes $\poly(n,m)$ time to compute the multiset $\sK$ of point stabilizers (one point stabilizer per orbit), and label $[m]$ by $\dot{\bigcup}_{K \in \sK} K \backslash M$, the right cosets in $M$ of the subgroups in $\sK$.  Compute the  multiset $\sum_{L \in \sL} [ \sF_L ]$ in $\poly(n,m, \norm{\sK}_1)$-time, by calling the $\SSF$ oracle.

 By Theorem~\ref{thm:HE comb main}, $[\sK] = \sum_{L \in \sL} [ \sF_L ]$. Via at most $m^2$ $\poly(n,m)$-time conjugacy checks between subgroups in $M$, compute the map  $\pi: \sK \leftrightarrow \sum_{L \in \sL} \sF_L $ that identifies conjugate subgroups. Compute the conjugating element for each pair. 
 
For each $L \in \sL$, use the map $\zeta$ of Construction~\ref{def:extension map} to label $\Omega$ by right cosets of elements in $\sL$. Define $\exthom$ by its natural action on cosets. 

\section{Reducing to \TriOrMultiSS}
\label{section:uniqueness}

In this section we prove Theorem~\ref{thm:main-methods-trireduction}, i.e., an instance $\psi$ of $\HomExtSym$ satisfying the conditions of Theorem~\ref{thm:main} will reduce to an instance $\OMS_\psi$ of $\TriOrMultiSS$. 

Fix an instance $\psi: M \to S_m$ of $\HomExtSym$ that satisfies the conditions of Theorem~\ref{thm:main}, i.e., $M = A_n$, $\ind{G}{M} = \poly(n)$ and $m < 2^{n-1}/\sqrt{n}$. Consider the instance $\OMS_\psi$ of$\OrMultiSS$ found via the reduction of Section~\ref{section:comb}. We will show that $\OMS_\psi$ satisfies the additional assumptions of $\TriOrMultiSS$ and provide answers for the additional oracles.  

\paragraph{Ordering, the $\preccurlyeq$ oracle\\}
The ordering $\preccurlyeq$ on conjugacy classes in $\cU = \ConjLeq(M)$ is given by ordering the indices of a representative subgroup for each conjugacy class. In other words, $[K_1] \preccurlyeq [K_2]$ if $\ind{M}{K_1} \leq \ind{M}{K_2}$. This relation is well-defined as conjugate subgroups have the same index. The relation $\preccurlyeq$ is clearly a total preorder. 

$\preccurlyeq$ oracle: The index of a subgroup $K \leq M$ can be computed in $\poly(n)$-time by Proposition~\ref{prop:pga-basic}. The $\preccurlyeq$ oracle compares two conjugacy classes in $\ConjLeq(M)$ by comparing the indices of two representatives.

\paragraph{Triangular condition, the $\triangle$ oracle\\ }

Here we define the $\triangle$ oracle on $\cU = \ConjLeq(M)$ (Construction~\ref{def:triangle-oracle}), analyze its efficiency (Remark~\ref{rmk:triangle-oracle-efficiency}), then prove its correctness (Lemma~\ref{lemma:HE L is unique if K is large}). The assumptions of Theorem~\ref{thm:main} are essential. 

First we set up some notation. By the assumptions of Theorem~\ref{thm:main}, $G = A_n$ and $M \leq G$ satisfies $\ind{G}{M} = \poly(n)$. Assume more specifically that $\ind{G}{M} < {n \choose r}$, for constant $r$. By Jordan-Liebeck (Theorem~\ref{thm:JordanLiebeck}) we find that $(A_n)_{(\Sigma)} \leq M \leq (A_n)_{\Sigma}$ for  some $\Sigma \subseteq [n]$ with $\lvert \Sigma \rvert < r$. Fix this subset $\Sigma \subset [n]$. 

Recall that, for a subset $\Sigma \subseteq [n]$ that is invariant under action by the permutation group $M \leq S_n$, we denote by $M^\Sigma \leq \Sym(\Sigma)$ the induced permutation group of the $M$-action on $\Sigma$. \\

\begin{construction}[$\triangle$ oracle]
\label{def:triangle-oracle}
We define a map $\triangle: \SubLeq(M) \to \SubLeq(G)$.\footnote{Though the $\triangle$ oracle returns an element of $\ConjLeq(G)$ on an input from $\ConjLeq(M)$, these conjugacy classes are represented by subgroups. So, the $\triangle$ oracle should return an element of $\SubLeq(G)$ on an input from $\SubLeq(M)$, while respecting conjugacy.}
Let $K \in \SubLeq(M)$. By Jordan-Liebeck, we find that $(A_n)_{(\Gamma)} \leq K \leq (A_n)_{\Gamma}$ for $\Gamma \subseteq [n]$ with $\abs{\Gamma}  < n/2$. There are two cases. If there is a subset $\Sigma_0 \subseteq \Gamma$ such that $K^{\Sigma_0} = M^\Sigma$, then let $\bar{\Gamma} = \Gamma \setminus \Sigma_0$ and 
	\begin{equation}
	\triangle(K) = \begin{cases} 
	\Alt([n] \setminus \bar{\Gamma}) \times K^{\bar{\Gamma}} & \text{if $K^{\bar{\Gamma}}$ is even} \\ 
	\text{the subgroup of index $2$ in } 
	\Sym([n] \setminus \bar{\Gamma}) \times K^{\bar{\Gamma}} & \text{if $K^{\bar{\Gamma}}$ contains an odd permutation} 
	\end{cases}.
	\end{equation}
If such a $\Sigma_0$ does not exist, then let $\triangle(K) = \textsf{Error}$. 
\end{construction}

\begin{remark}[Efficiency of $\triangle$ oracle] 
\label{rmk:triangle-oracle-efficiency}
Answering the $\triangle$ oracle of Construction~ \ref{def:triangle-oracle} requires finding orbits, finding the induced action on orbits, and checking permutation equivalence (or conjugacy of point stabilizers, per Corollary~\ref{cor:prelim-equiv-vs-conj}). These can be accomplished in $\poly(n, m)$ time (Propositions~\ref{prop:pga-basic} and~\ref{prop:pga-subgroup}).
\end{remark}

\begin{remark}
The $\triangle$ oracle is well-defined as a $\ConjLeq(M) \to \ConjLeq(G)$ map. 
\end{remark}
Now, we prove that the oracle $\triangle$ (Definition~\ref{def:triangle-oracle}) satisfies the conditions of $\TriOrMultiSS$. 
In other words, the equivalence class of the $M$-action on its longest orbit uniquely determines the equivalence class of the transitive $G$-action and this correspondence is injective.  Lemma~\ref{lemma:HE L is unique if K is large} makes this more precise. 
\begin{lemma}
\label{lemma:HE L is unique if K is large}
Let $M \leq G = A_n$ have index $\ind{G}{M} \leq {n \choose u}$. Let $G$ act on $\Omega$ transitively, with degree $\abs{\Omega} < {n \choose v}$. Assume $u + v < n/2$. 
If  $K_0$ is a point stabilizer of the induced $M$ action on its longest orbit, then $\triangle(K_0)$ is a point stabilizer of the $G$-action on $\Omega$.\footnote{If $M$ and $K_0$ are known, then $\triangle(K_0)$ is uniquely determined.}
\end{lemma}

To rephrase, if $M$ acts on its longest orbit as an $(M,K_0)$-action, then $G$ acts as a $(G,\triangle(K_0))$-action. 

We defer the proof of Lemma~\ref{lemma:HE L is unique if K is large} to present a few useful claims. 

\begin{claim}
\label{claim:HE (Delta,LDelta) determines large L}
If $(A_n)_{(\Sigma)} \leq L \leq (A_n)_{\Sigma}$, then the pair $(\Sigma, L^\Sigma)$ determines $L$. 
\end{claim}
\begin{proof}

We have two cases. Either $L = (A_n)_{(\Sigma)} \times L^\Sigma = A_{n-\abs{\Sigma}} \times L^\Sigma$, or $L$ is an index $2$ subgroup of $(S_n)_{(\Sigma)} \times L^\Sigma = S_{n-\abs{\Sigma}} \times L^\Sigma$. In the first case, all permutations in $L^\Sigma$ must be even. In the second case, $L^\Sigma$ must contain an odd permutation. 
\end{proof}

\begin{claim}
\label{claim:HE LDelta = McapL Delta} 
Suppose that $(A_n)_{(\Sigma)} \leq L \leq (A_n)_{\Sigma}$ and $(A_n)_{(\Gamma)} \leq M \leq (A_n)_{\Gamma}$ for $\Gamma \cap \Sigma = \emptyset$. Then, $L^\Sigma = (L \cap M)^\Sigma$. (Equivalently, $M^\Gamma = (L \cap M)^\Gamma$.) 
\end{claim}
\begin{proof}
The inclusion $\supseteq$ is obvious. We show $\subseteq$. 

Let $\sigma \in L^\Sigma$. View $\sigma$ as a permutation in $S_n$. Let $\Sigma \subseteq  [n]$ be such that $[n] = \Gamma \dotcup \Sigma \dotcup \Sigma$. Consider the set $T = \{ \tau \in S_n : \supp(\tau) \subseteq \Sigma \text{ and } \sgn \tau = \sgn \sigma \}.$

We see that for all $\tau \in T$, $\sigma \tau \in M \cap L$. Thus, $\sigma \in (M \cap L)^\Gamma$. 
\end{proof}

\begin{proof}[Proof of Lemma~\ref{lemma:HE L is unique if K is large}]

Let $L$ be a point stabilizer of $G$ acting on $\Omega$. Since $\abs{\Omega} < {n \choose v}$, by Jordan-Liebeck Theorem~\ref{thm:JordanLiebeck}, there exists a subset $\bar{\Gamma} \subset [n]$ such that $(A_n)_{(\bar{\Gamma})} \leq L \leq (A_n)_{\bar{\Gamma}}$ and $\abs{\bar{\Gamma}} < v$. Similarly, there exists $\Sigma \subset [n]$ such that $(A_n)_{(\Sigma)} \leq M \leq (A_n)_{\Sigma}$ and $\abs{\Sigma} < u$. Fix $\bar{\Gamma}$ and $\Sigma$. 

By Theorem~\ref{thm:HE comb main}, we find that the point stabilizers of the $M$-action on $\Omega$ are described by $\sF_L$. By Definition~\ref{def:FsubL} and Corollary~\ref{cor:similar-multiset-same-action}, we find that
	\begin{equation*}	
	K_0 = \argmax \{\ind{M}{K} : K \in \sF_L \} \sim_M \argmin \{ \abs{K}: K = g^{-1} L g \cap M \text{ for } g \in G \}.  
	\end{equation*}
But, $\abs{g^{-1} L g \cap M}$ is minimized when $g \in G = A_n$ satisfies $\Gamma^g \cap \Sigma = \emptyset$.  Fix this $g$.  By Claims~\ref{claim:HE (Delta,LDelta) determines large L} and~\ref{claim:HE LDelta = McapL Delta} applied to $g^{-1} L g$ and $M$, we find that 

\begin{equation}
g^{-1} L g =  \begin{cases} 
	\Alt([n] \setminus \bar{\Gamma}) \times K^{\bar{\Gamma}} & \text{if $K^{\bar{\Gamma}}$ is even} \\ 
	\text{the subgroup of index $2$ in } 
	\Sym([n] \setminus \bar{\Gamma}) \times K^{\bar{\Gamma}} & \text{if $K^{\bar{\Gamma}}$ contains an odd permutation} 
\end{cases}.
\end{equation}
In other words, we have found that $g^{-1} L g = \triangle(K_0)$, i.e., $L \sim_G \triangle(K_0)$. 
It follows that the $G$-action on $\Omega$ is a $(G, \triangle(K_0))$-action. 
\end{proof}

\section{Generating extensions within one equivalence class}
\label{section:within eq class}
We now consider how to, given one extension $\exthom \in \Hom(G, S_m)$ of $\psi \in \Hom(M, S_m)$, generate all extensions of $\psi$ equivalent to $\exthom$.

\begin{theorem}
\label{thm:within eq class}
Let $M \leq G$ and $\psi \in \Hom(M, S_m)$. Suppose that $\exthom \in \Hom(G,S_m)$ extends $\psi$. Then the class of extensions equivalent to $\exthom$ can be efficiently enumerated. 
\end{theorem}

We will see that proving this result reduces to finding coset representatives for subgroups of permutation groups. First, some notation for describing group actions equivalent to $\exthom$. 

\begin{notation}
Let $\lambda \in S_m$. Let $\exthom \in \Hom(G,S_m)$. Define $\exthom^\lambda \in\Hom(G,S_m)$ by $\exthom^\lambda(g) = \lambda^{-1} \exthom(g) \lambda$ for all $g \in G$. 
\end{notation}

While $\exthom^\lambda$ will be equivalent to $\exthom$, regardless of the choice of $\lambda \in S_m$, we remark on the distinction between $\exthom^\lambda$ being the same group action, an equivalent extension of $\psi$, and an equivalent action.  

\begin{remark}
Let $\lambda \in S_m$. Let $\varphi_1, \varphi_2 \in \Hom(G,H)$. 
\begin{itemize}
\item $\varphi_1$ and $\varphi_2$ are equivalent (as a permutation actions) $\iff$ $\varphi_1 = \varphi_2^\lambda$ for some $\lambda \in S_m$. 
\item $\varphi_1$ and $\varphi_2$ are equivalent extensions of $\psi$ $\iff$ $\varphi_1 = \varphi_2^\lambda$ and $\varphi_1|_M = \psi$ $\iff$ $\varphi_1 = \varphi_2^\lambda$ for some $\lambda \in C_{S_m}(\exthom_1(M)) = C_{S_m}(\psi(M))$. 
\item $\varphi_1$ and $\varphi_2$ are equal $\iff$ $\varphi_1 = \varphi_2^\lambda$ for some $\lambda \in C_{S_m}(\varphi_1(G))$. 
\end{itemize}
\end{remark}

We conclude that the sets of coset representatives of $C_{S_m}(\exthom(G))$ in $C_{S_m}(\psi(M))$ generate the non-equal equivalent extensions of $\psi$. 

\begin{remark}
\label{rmk:eqclass-vs-cosetreps}
Let $R$ be a set of coset representatives of $C_{S_m}(\exthom(G))$ in $C_{S_m}(\psi(M))$. The set of equivalent extensions to $\exthom$ can be described (completely and without repetitions) by 
$$
	\{ \exthom^\lambda: \lambda \in R \} .$$
\end{remark}

These centralizers can be found in $\poly(n, m)$-time. The centralizer of a set of $T$ permutations in $S_m$ can be found in $\poly(\abs{T}, m)$ time (see Section~\ref{section:appendix-centralizer}), and we use this with the set of generators of $M$ and $G$. We can now apply the cited unpublished result by Blaha and Luks, stated below and proved in Section~\ref{section:luks}.

\begin{theorem}[Blaha--Luks]   \label{thm:luks}
Given subgroups $K\le L\le S_m$, one can efficiently enumerate a representative of each coset of $K$ in $L$.
\end{theorem}

Since coset representatives of $K = C_{S_m}(\psi(M))$ in $L = C_{S_m}(\exthom(G))$ can be efficiently enumerated, so can all equivalent extensions to $\exthom$, by Remark~\ref{rmk:eqclass-vs-cosetreps}.

As a corollary, we find that the number of equivalent extensions can be computed in $\poly(n,m)$ time. 
\begin{corollary}
Suppose $\exthom \in \Hom(G,S_m)$ extends $\psi \in \Hom(M,S_m)$. The number of equivalent extensions to $\exthom$ is  $\ind{C_{S_m}(\exthom(G)) }{C_{S_m}(\psi(M) }$. This can be computed in $\poly(n,m)$-time. 
\end{corollary}

\section{Integer linear programming for large $m$}
\label{section:large}
There is an interesting phenomenon for very large $m$, when $m > 2^{1.7^{n^2}}$. The instances $\OMS_\psi$ of $\OrMultiSS$ can be solved in polynomial time.

$\MultiSS$ can naturally be formulated as an \textsc{Integer Linear Program}, with dimensions $\abs{\cU} \times \abs{\cV}$, the size of the universe $\cU$ and length of the list $\SSF$ (indexed by $\cV$). The variables correspond to multiplicities of the elements of $\SSF$. The constraints correspond to elements of $\cU$, by checking whether their multiplicities in the multiset and subset sum are equal. 

 In $\OMS_\psi$, these are $\Conj(M)$ and $\Conj(G)$. A result of Pyber~\cite{PyberSubgroupsSn} says that for $G \leq S_n$, the number of of subgroups is bounded by $\abs{\Sub(S_n)} \leq 1.69^{n^2}$. This bound is tight, so we cannot hope for the number of variables $(\Conj(M))$ to be smaller than exponential in $n^2$. 

 The ``low-dimensional'' algorithms of Lenstra and Kannan solve \textsc{Integer Linear Programming} in ``polynomial'' time \cite{Lenstra1983, KannanIP}, which are sufficient for this purpose. We state their results more precisely below. 

\begin{theorem}
\label{thm:prelim KannanIP}
The \textsc{Integer Linear Programming}--Search and Decision Problems can be solved in time $N^{O(N)} \cdot s$, where $N$ refers to the number of variables and $s$ refers to the length of the input.\footnote{This result shows that {\ILP} is fixed-parameter tractable, but we will not use that terminology here.}  
\end{theorem}

\begin{lemma}
Suppose that the \textsc{Integer Linear Programming} Search Problem can be solved in time $f(N, M,a)$. Then, the \textsc{Integer Linear Programming} Threshold-$k$ Enumeration Problem can be solved in time $f(N,M, a) \cdot O(k^2)$. 
\end{lemma}

We have found that, for instances $\psi$ of $\HomExtSym$ with $m > 2^{1.7^{n^2}}$, the Threshold-$k$ Enumeration Problem for $\OMS_\psi$ can be solved in $\poly(n,m, k)$-time.  For these instances of $\psi$, the Threshold-$k$ Enumeration Problem can be solved in $\poly(n,m,k)$-time.

\section{Background: permutation group algorithms}
\label{section:pga}

\subsection{Basic results}

We present results we use from the literature on permutation group algorithms.
Our main reference is the monograph~\cite{SeressPGA}. 

Recall that a group $G$ is \defn{given} or \defn{known} when a set of generators for $G$ is given/known. A coset $Ga$ is \defn{given} or \defn{known} if the group $G$ and a coset representative $a' \in Ga$ are given/known. A group (or a coset) is \defn{recognizable} if we have an oracle for membership and \defn{recognizable in time $t$} if the membership oracle can be implemented in time $t$. 

\begin{proposition}
\label{prop:pga-membership-test}
Membership in a given group $G \leq S_n$ (or coset $Ga$) can be tested in $\poly(n)$ time. In other words, a known group (or coset) is polynomial-time recognizable. 
\end{proposition}
\begin{proof}
This is accomplished by the Schreier-Sims algorithm, see \cite[Section 3.1 item (b)]{SeressPGA}.
\end{proof}

\begin{corollary}
\label{cor:pga-intersection-recognizable}
If $G_1, \ldots, G_k \leq S_n$ and $a_1, \ldots, a_k \in S_n$ are given, then the intersection $\bigcap_{i} G_i a_i$ is polynomial-time recognizable. 
\end{corollary}

\begin{proposition}
\label{prop:pga-basic}
Given $G \leq S_n$, the following can be computed in $\poly(n)$-time. 

\begin{enumerate}[(a)]
\item A set of $\leq 2n$ generators of $G$. \label{item:pga-nonredundant-generators} 
\item The order of $G$.  
	\label{item:pga-order}
\item The index $\ind{G}{M}$, for a given subgroup $M \leq G$. 
	\label{item:pga-index}
\item The orbits of $G$. \label{item:pga-orbits} 
\item The point stabilizers of $G$. 
\end{enumerate}
\end{proposition}
\begin{proof}
Most items below are addressed in \cite[Section 3.1]{SeressPGA}. 
\begin{enumerate}[(a)]
\item Denote by $T$ the set of given generators of $G$. Use membership testing to prune $T$ down to a non-redundant set of generators. By~\cite{Bab_subgroupchain}, the length of subgroup chains in $S_n$ is bounded by $2n$, so $\abs{T}  \leq 2 n$ after pruning.
\item See \cite[Section 3.1 item (c)]{SeressPGA}.
\item Compute $\abs{M}$ and $\abs{G}$.
\item See \cite[Section 3.1 item (a)]{SeressPGA}. 
\item See \cite[Section 3.1 item (e)]{SeressPGA}.
\end{enumerate}
\end{proof}

\begin{proposition}
\label{prop:pga-recognizable-to-given}
Let $M \leq G$ be a recognizable subgroup of $G$ of index $\ind{G}{M} = s$. A set of generators for $M$ and a set of coset representatives for $M \backslash G$ can be found in $\poly(n,s)$ time (including calls to the membership oracle). 
\end{proposition}
\begin{proof}
Consider the subgroup chain $G \geq M \geq M_1 \geq M_{(12)} \geq M_{(123)} \geq M_{(12\cdots n)} = 1$ ($M$ is followed by its stabilizer chain). Apply Schreier-Sims to this chain. (This is the ``tower of groups'' method introduced in~\cite{BabaiLV} and derandomized in~\cite{FurstHopcroftLuks}. Note that this method only requires the subgroups in this chain to be recognizable.)
\end{proof}

\begin{proposition}
\label{prop:pga-subgroup}
Let $G \leq S_n$ be a given permutation group. Let $M, L \leq G$ be given subgroups. Denote their indices by $s = \ind{G}{M}$ and $t = \ind{G}{L}$. 
\begin{enumerate}[(a)]
\item The normalizer  $N_G(M)$ can be found in $\poly(n,s)$-time. 
\item The number of conjugates of $M$ in $G$ can be computed in $\poly(n,s)$-time. 
\item The conjugacy of $L$ and $M$ in $G$ can be decided and a conjugating element $g \in G$ such that $g^{-1} L g = M$ can be found if it exists, in $\poly(n, s)$ time. 
\end{enumerate}
\end{proposition}
\begin{proof}
\begin{enumerate}[(a)]
\item Let $S$ be the given set of generators of $M$. Take a set of coset representatives for $M \backslash G$, found by Proposition~\ref{prop:pga-recognizable-to-given}. Remove the coset representatives $g$ that do not satisfy $g^{-1} S g  \subseteq M$.  This is accomplished through membership testing. The remaining coset representatives, along with $S$, generate  $N_G(M)$. 

\item The number of conjugates of $M$ in $G$ is the index $\ind{G}{N_G(M)}$. 

\item Check if $\abs{L} = \abs{M}$ by Proposition~\ref{prop:pga-basic} \ref{item:pga-order}. If not, they are not conjugate. Otherwise, let $S$ be the set of  given generators of $M$. Now, $L$ and $M$ are conjugate if and only if there exists a coset representative $g$ for $N_G(M) \backslash G$ that satisfies $g^{-1} S g \subseteq L$. 
\end{enumerate}\end{proof}

\begin{proposition}
\label{prop:pga-double-cosets}
Let $G \leq S_n$ be a given permutation group. Let $M, L \leq G$ be given subgroups. Denote their indices by $s = \ind{G}{M}$ and $t = \ind{G}{L}$. 
\begin{enumerate}[(a)]
\item Given $g, h \in G$, membership of $h$ in the double coset $LgM$ can be  decided in $\poly(n, \min\{s,t\})$-time. 
\item A set of double coset representatives for $L \backslash G /M$ can be found in $\poly(n, \min\{s,t\})$-time. 
\end{enumerate}
\end{proposition}
\begin{proof}
\begin{enumerate}[(a)]
\item Without loss of generality assume that $s \leq t$. Notice that 	
	\begin{align*}
	h \in LgM & \iff Lh \cap g M \neq \emptyset 
	 \iff g^{-1} L h \cap M \neq \emptyset 
	\iff (g^{-1}Lg)  \cap Mh^{-1} g \neq \emptyset. 
	\end{align*}
So, deciding whether $h \in LgM$ is equivalent to deciding whether the subgroup $L^* = g^{-1} L g$ and coset $ M g^*$ have non-empty intersection, where $g^* = h^{-1} g$. This intersection, $L^* \cap M g^*$, is either empty or a right coset of $L^* \cap M$ in $L^*$. In what remains we check whether a coset of $L^*\cap M$ is contained in $L^* \cap M g^*$. 

Notice that  $\ind{L^*}{L^* \cap M} \leq \ind{G}{M}  = s$. Find a set $R$ of coset representatives  of $L^* \cap M$ in $L^*$ using Proposition~\ref{prop:pga-recognizable-to-given}, noting that $L^*\cap M$ is recognizable (Corollary~\ref{cor:pga-intersection-recognizable}). 
For each representative $r \in R$, check whether $r \in L^* \cap Mg^*$ (Corollary~\ref{cor:pga-intersection-recognizable}).

\item A list of $t$ coset representatives of $M$ in $G$ is a redundant set of double coset representatives for $L \backslash G /M$. This can be pared down to a set of non-redundant double coset representatives by ${t \choose 2}$ comparisons using part (a).
\end{enumerate}
\end{proof}

\subsection{Generators and relations}
\label{section:pga-generators-relations}

Let $x_1, \ldots, x_s$ be free generators of the free group $F_s$. Let $R_1, \ldots, R_t \in F_q$. The notation $G = \langle x_1, \ldots, x_s \mid R_1, \ldots, R_t\rangle$ refers to the group $F_s/N$ where $N$ is the normal closure of $\{R_1, \ldots, R_t\}$. This notation is referred to as a generator-relator presentation of $G$; the $R_i$ are called the relators. 

\begin{definition}[Straight-line program]
Let $X$ be a set of generators of a group $H$. A \defn{straight line program} in $H$ starting from $X$ reaching a subset $Y \subseteq H$ is a sequence $h_1, \ldots, h_m$ of elements of $H$ such that, for each $i$, either $h_i \in S$, or $h_i^{-1} \in S$, or $(\exists j, k < i)(h_i = h_jh_k)$, and $Y \subseteq \{h_1, \ldots, h_m\}$.
\end{definition}

We shall say that a straight line program is \defn{short} if its length is $\poly(n)$, where $n$ is a given input parameter. 

\begin{theorem}
\label{thm:pga-straight-line}
Let $G \leq S_n$ given by a set $S=\{a_1, \ldots, a_s\}$ of generators. Then, there exists a presentation $G \cong \langle x_1, \ldots, x_s \mid R_1, \ldots, R_t\rangle$ such that the set $\{R_1, \ldots, R_t\}$ is described by a short straight-line program, and the free generator $x_i$ corresponds to $a_i$ under the $F_s \to G$ epimorphism. Moreover, this straight-line program can be constructed in polynomial time. 
\end{theorem}

The proof of this well-known fact follows from the Schreier-Sims algorithm. 

\subsection{Extending a homomorphism from generators}
\label{section:promise}

We address Remark~\ref{rmk:promise} that $\HomExtSym$ is not a promise problem. The input homomorphism $\psi: M \to H$ is represented by its values on generators of $M$. Whether this input does indeed represent a homomorphism, i.e., whether the values on the generators extend to a homomorphism on $M$, can be verified in $\poly(n)$ time. 

\begin{proposition}
Let $G\leq S_n$ and $H \leq S_m$ be permutation groups. Let $S=\{a_1, \ldots, a_s\}$ be a set of generators of $G$ and $f: S \to H$ a function. Whether $f$ extends to a $G \to H$ homomorphism is testable in $\poly(n,m)$ time. 
\end{proposition}
\begin{proof}
By Theorem~\ref{thm:pga-straight-line}, a generator-relator presentation of $G$ can be found in $\poly(n)$ time, in the sense that the relators are described by straight-line programs constructed in $\poly(n)$ time. If $R_i(a_1, \ldots, a_s)$ is one of the relators, then we can verify $R_i(f(a_1), \ldots, f(a_s)) =1$ in time $\poly(n,m)$ by evaluating the straight-line program. The validity of these equations is necessary and sufficient for the extendability of $f$. 
\end{proof}

In particular, whether inputs to $\HomExtSym$ satisfy the conditions of Theorems~\ref{thm:bounded}--\ref{thm:large-m} (and Theorems~\ref{thm:bounded-enum}--\ref{thm:large-enum}) can be verified in $\poly(n)$ time.


\subsection{Centralizers in $S_n$}
\label{section:appendix-centralizer}

\begin{proposition}
Given $G \leq S_n$, its centralizer $C_{S_n}(G)$ in the full symmetric group can be found in polynomial time. 
\end{proposition}
\begin{proof}
Let $T = \{t_i\}_i$ denote the given set of generators for $G$. Without loss of generality, we may assume $\abs{T} \leq 2 n$ by Proposition~\ref{prop:pga-basic}~\ref{item:pga-nonredundant-generators}. 

Construct the permutation graph $X = (V,E)$ of $G$, a colored graph on vertex set $V = [n]$ and edge set $E = \bigcup_{t \in T}E_t$, where $E_t= \{ (i, i^t): i \in [n] \}$ for each color $t \in T$. The edge set colored by $t \in T$ describes the permutation action of $t$ on $[n]$. We see that $C_{S_n}(G) = \Aut(X)$, where automorphisms preserve color by definition. 

If $G$ is transitive ($X$ is connected), then $C_{S_n}(G)$ is semiregular (all point stabilizers are the identity). For $i, j \in [n]$, it is possible in $\poly(n)$ time to decide whether there exists a permutation $\sigma \in \Aut(G) = C_{S_n}(G)$ satisfying $i^\sigma = j$ (takes $i$ to $j$), then find the unique $\sigma$ if it exists. To see this, build the permutation $\sigma$ by setting $i^\sigma = j$, then following all colored edges from $i$ and $j$ in pairs to assign $\sigma$. If this is a well-defined assignment, then the permutation $\sigma \in \Aut(X)$ satisfying $i^\sigma = j$ exists. 

In fact, if $X_1 = (V_1, E_1)$ and $X_2 = (V_2, E_2)$ are connected, whether then a graph isomorphism taking $i \in V_1$ to $j \in V_2$ can be found in $\poly(\abs{V_1})$ time if one exists. 

If $X$ is disconnected, collect the connected components of $X$ by isomorphism type, so that there are $m_i$ copies of the connected graph $X_i$ in $X$, where $i = 1 \ldots \ell$ numbers the isomorphism types. The components and multiplicities can be found in $\poly(n)$ time by finding the components of $X$ (or, orbits of $G$, by Proposition~\ref{prop:pga-basic}~\ref{item:pga-orbits}) and pairwise checking for isomorphism. The automorphism group of $X$ is 
	\begin{equation*}
	\Aut(X) = \Aut(X_1) \wr S_{m_1} \times \cdots \times \Aut(X_\ell) \wr S_{m_\ell}. 
	\end{equation*}
Each $X_i$ is  connected, so $\Aut(X_i)$ can be found as above. 
\end{proof}

\section{Blaha-Luks: enumerating coset representatives }  \label{section:luks}
We sketch the proof of the unpublished result by Blaha and Luks (Theorem~\ref{thm:luks}), restated here for convenience. Below, by ``coset'' we mean ``right coset.'' 

\begin{theorem}[Blaha--Luks]
\label{thm:luks-local}
Given subgroups $K\le L\le S_n$, one can efficiently enumerate (at  $\poly(n)$ cost per item) a representative of each coset of $K$ in $L$.
\end{theorem}

Let $\MoveCoset(M\sigma, i, j)$ be a routine that decides whether there exists a permutation $\pi \in M\sigma$ satisfying $i^\pi = j$, and if so, finds one.
\begin{proposition}
$\MoveCoset$ can be implemented in polynomial time. 
\end{proposition}
\begin{proof}
Answering $\MoveCoset$ is equivalent to finding $\pi \in M$ satisfying $i^\pi = j^{\sigma^{-1}}$ if one exists. 
This is the same as finding the orbits of $M$ (Proposition~\ref{prop:pga-basic}~\ref{item:pga-orbits}).
\end{proof}

\begin{definition}[Lexicographic ordering of $S_n$]
Let us encode the permutation $\pi \in S_n$ by the string $\pi(1) \pi(2) \cdots \pi(n)$ of length $n$ over the alphabet $[n]$. 
Order permutations lexicographically by this code. 
\end{definition}
 Note that the identity is the lex-first permutation in $S_n$. 
\begin{lemma}
Let $\sigma \in S_n$ and $K \leq S_n$. The algorithm \LexFirst (below) finds the
lex-first element of the subcoset $K \sigma\subseteq S_n$ in polynomial time. 
\end{lemma}

\begin{algorithm}[H]       
\caption{LexFirst within Subcoset}
\label{algorithm:LexFirst}
\begin{algorithmic}[1]
            \Procedure{LexFirst}{subcoset $K\sigma$}
            \State \textbf{for} $ i \in [n]$ \textbf{do} \; $i^\pi \gets \Null$\pcom{ Initialize $\pi:[n] \rightarrow [n]\cup\{\Null\}$}
            \For{$s \in [n]$} \pcom{Find smallest possible image of $1$ under a permutation in $K\sigma$, then iterate.}
                
                \For{$t \in [n]$}\pcom{Find smallest $s^\pi$ possible by checking $[n]$ in order}
                	\State \textbf{if} $\MoveCoset(K\sigma,s,t) = \True$ \textbf{break}
                \EndFor
                \State $s^\pi \gets t$
               	\State $\tau \gets \MoveCoset(K\sigma, s, t)$ \pcom{Restrict subcoset to elements moving $s$ to $t$}
            \EndFor
            \State \Return $\pi$
       	\EndProcedure
\end{algorithmic}
\end{algorithm}

It is straightforward to verify the correctness and efficiency of \LexFirst. \qed

\begin{proof}[Proof of Theorem~\ref{thm:luks-local}]
Let $K \leq L \leq S_n$. Let $S$ be a set of generators of $L$. 
The \defn{Schreier graph} $\Gamma =\Gamma(K \backslash L, S)$ is the permutation graph of the $L$-action on the coset space $K \backslash L$, with respect to the set $S$ of generators. $\Gamma$ is a directed graph with  vertex set $V = K \backslash L$ and edge set $E = \{ (i, i^\pi): i \in [n], \pi \in S  \}$. 

To prove Theorem~\ref{thm:luks-local}, we may assume $\abs{S} \leq 2n$, by Proposition~\ref{prop:pga-basic}\ref{item:pga-nonredundant-generators}. Use breadth-first search on  $\Gamma$, constructing $\Gamma$ along the way. Represent each vertex (a coset) by its lexicographic leader. Then, store the discovered vertices, ordered lexicographically, in a balanced dynamic search tree such as a red-black tree. Note that the tree will have $O(\log(n!)) = O(n \log n)$ depth and every vertex of $\Gamma$ has at most $2n$ out-neighbors. Hence, the incremental cost is $\poly(n)$. 
\end{proof}

\section{List-decoding motivation for $\HomExt$ Search and Threshold-$k$ Enumeration} 
\label{section:appendix-motivation}


In this appendix we shall (a) indicate that  Homomorphism Extension is a natural component of list-decoding homomorphism codes, (b) discuss the role of Theorem~\ref{thm:main} in list-decoding, and (c) motivate the special role of Threshold-$2$ Enumeration in this process. We note that all essential ideas in $\HomExt$ Threshold-$k$ Enumeration already occur in the Threshold-$2$ case. 

A function $\psi: G \to H$ is an \defn{affine homomorphism} if $\vf(a b^{-1} c) = \vf(a) \vf(b)^{-1} \vf(c)$ for all $a, b, c \in G$, or, equivalently, if $\vf = h_0 \cdot \vf_0$ for an element $h_0 \in H$ and homomorphism $\vf_0: G \to H$. For groups $G$ and $H$, let $\aHom(G,H)$ denote the set of affine $G \to H$ homomorphisms. 
Let $H^G$ denote the set of all functions $f: G \to H$. We view $\aHom(G,H)$ as a (nonlinear) code within the code space $H^G$ (the space of possible ``received words'') and refer to this class of codes as \defn{homomorphism codes}. ($H$ is the alphabet.) These codes are candidates for \defn{local} list-decoding up to minimum distance. For more detailed motivation see~\cite{ GKS06,DGKS08, homcodes}.

In~\cite{homcodes}, the \textsc{Homomorphism Extension} Search Problem arises as a natural roadblock to list-decoding homomorphism codes, if the minimum distance does not behave nicely. 

To elaborate, the minimum distance of $\aHom(G,H)$ is the minimum normalized Hamming distance between two $G \to H$ affine homomorphisms. The complementary quantity is the \defn{maximum agreement}, which for the code $\aHom(G,H)$  we denote by 
\begin{equation}
\label{eqn:Lambda}
\Lambda = \Lambda_{G,H} = \max_{\substack{\vf_1, \vf_2 \in \aHom(G,H) \\ \vf_1 \neq \vf_2}} \agr(\vf_1, \vf_2), 
\end{equation} 
where $\agr(\vf_1, \vf_2) = \frac{1}{\abs{G}} \abs{\{ g \in G: \vf_1(g) = \vf_2(g) \}}$ is the fraction of inputs on which two homomorphisms agree. 

\paragraph{(a) $\HomExt$ as a component of list-decoding\\}
When list-decoding a function $f:G \to H$, i.e., finding all $\vf \in \aHom(G,H)$ satisfying  $\agr(f, \vf) \geq \Lambda + \epsilon$ for fixed $\epsilon >0$, we run into difficulty if there is 
a subgroup $M \lneq G$ satisfying $\abs{M} > (\Lambda + \epsilon)\abs{G}$. In this case, it is possible for the agreement between $f$ and $\psi$ to lie entirely within $M$. As a consequence, $f$ may only provide information on the restriction $\vf|_M: M \to H$ of $\vf$ to $M$, but not on its behavior outside $M$. The natural objects returned by our list-decoding efforts are such partial homomorphisms, defined only on the subgroup $M$. We see from this that solving $\textsc{Homomorphism Extension}$ from subgroups of density greater than $\Lambda$ is a natural component to full list-decoding.

Works prior to~\cite{homcodes} considered cases for which $\Lambda$ was known, so it could be guaranteed that affine homomorphisms $\vf$ in the output satisfied $\agr(f, \vf) > \Lambda + \epsilon/2$.\footnote{In $\poly(\delta,\log \abs{G})$ time, we can estimate $\agr(f, \vf)$ for $\vf$ in the output list to within $\delta$ with high confidence. With this, we can prune the small agreement homomorphisms satisfying $\agr(f, \vf) \leq \Lambda+\epsilon/2$ with high probability.\label{footnote:pruning}} 
 Additionally, they considered classes of groups for which defining an affine homomorphism on a set of density greater than $\Lambda$ immediately defined the affine homomorphism on the whole domain, so $\HomExt$ was not an issue. 
 
\paragraph{(b) The case $G$ is alternating, $M$ has polynomial index\\}
One of the main results stated in~\cite{homcodes} is the following. 
\begin{theorem}
\label{thm:homcodes-alttoalt}
Let $G =A_n$, $H = S_m$ and $m < 2^{n-1}/\sqrt{n}$. Then, $\aHom(G,H)$ is algorithmically list-decodable, i.e., there exists a list-decoder that decodes $\aHom(G,H)$ up to distance $(1-\Lambda-\epsilon)$ in time $\poly(n, m, 1/\epsilon)$ for all $\epsilon>0$. 
\end{theorem}
The proof of this result depends on the main result of the present paper, Theorem~\ref{thm:main}, in the following way. 

For $A_n$, the theory of permutation groups tells us that $\Lambda \geq 1/{n\choose 2}$. It depends on $H$ whether this lower bound is tight. What the algorithm in~\cite{homcodes} actually finds is an intermediate output list consisting of $M \to S_m$ homomorphisms, where $M\leq A_n$ has order greater than $\Lambda\abs{A_n}$, i.e., $\ind{A_n}{M} < {n \choose 2}$. Our Theorem~\ref{thm:main} solves $\HomExt$ for the case $\ind{A_n}{M} = \poly(n)$ and $m<2^{n-1}/\sqrt{n}$, completing the proof of Theorem~\ref{thm:homcodes-alttoalt}.

\begin{remark}
\label{rmk:roadblock-to-LD}
The restrictions on $H$ in Theorem~\ref{thm:homcodes-alttoalt} arise from the limitations of the $\HomExt$ results in this paper. Any $\HomExt$ results relaxing conditions on $H$ would automatically yield the same relaxations on $H$ for list-decoding, potentially extending the validity of all permutation groups $H$. In this sense, the \textit{limitations of our understanding of the Homomorphism Extension Problem constitute one of the main roadblocks to list-decoding homomorphism codes for broader classes of groups. }
\end{remark}

\paragraph{(c) Role of Threshold-$2$ Enumeration\\}

Our discussion above shows that $\Lambda$ is the lower threshold for densities of subgroups from which $\HomExt$ must extend.  Also, the algorithm of~\cite{homcodes} guarantees that only partial homomorphisms with domain density greater than $\Lambda$ are of interest. 

However, the actual value of $\Lambda$ is not obvious to compute, nor is it automatically given as part of the input to a list-decoding problem. Lower bounds on $\Lambda$ are necessary to make $\HomExt$ tractable; they also improve the algorithmic efficiency and output quality in list-decoding. Solving $\HomExt$ Threshold-2 Enumeration instead of $\HomExt$ Search, when extending lists of partial homomorphisms, can provide (or improve) lower bounds on $\Lambda$. 

It is easy to see how Threshold-$2$ helps improve our lower bound on $\Lambda$. If a partial homomorphism $\psi$ extends non-uniquely, $\HomExt$ Threshold-$2$ returns a pair of homomomorphisms whose agreement is larger than the domain of $\psi$. So, their agreement (and the density of the domain of $\psi$) gives witness to an updated lower bound on $\Lambda$. 

Better lower  bounds for  $\Lambda$ have three main consequences. 
\begin{itemize}
\item As discussed, better lower bounds for $\Lambda$ relax the requirements for the $\HomExt$ algorithm called by the list-decoder. It suffices to extend from subgroups with densities above the lower bound. 
\item Since the algorithm of~\cite{homcodes} guarantees that only partial homomorphisms with domain density greater than $\Lambda$ are of interest, the intermediate list of partial homomorphisms may be pruned. 
\item Once a list of full homomorphisms is generated, a better lower bound allows better pruning of the output list of a list-decoder 
(discussed in footnote~\ref{footnote:pruning}).  
\end{itemize}

\bibliographystyle{alpha}
\bibliography{bibHE}

\newcommand{\noop}[1]{}
\begin{thebibliography}{DGKS08}

\bibitem[Bab79]{BabaiLV}
L\'aszl\'o Babai.
\newblock Monte-{C}arlo algorithms in graph isomorphism testing.
\newblock {\em Universit{\'e} tde Montr{\'e}al Technical Report, DMS}, pages
  79--10, 1979.

\bibitem[Bab86]{Bab_subgroupchain}
L\'aszl\'o Babai.
\newblock On the length of subgroup chains in the symmetric group.
\newblock {\em Communications in Algebra}, 14(9):1729--1736, 1986.

\bibitem[BBW18]{homcodes}
L\'{a}szl\'{o} Babai, Timothy Black, and Angela Wuu.
\newblock List-decoding homomorphism codes with arbitrary codomain.
\newblock In {\em APPROX-RANDOM}, 2018.
\newblock To appear.

\bibitem[BL94]{BlahaLuks}
Ken Blaha and Eugene~M. Luks.
\newblock P-complete permutation group problems.
\newblock In {\em Proc. 25th Southeastern Conf. on Combinatorics, Graph Theory,
  and Computing}, volume 100 of {\em Congressus Numerantium}, pages 119--124,
  1994.

\bibitem[DGKS08]{DGKS08}
Irit Dinur, Elena Grigorescu, Swastik Kopparty, and Madhu Sudan.
\newblock Decodability of group homomorphisms beyond the {J}ohnson bound.
\newblock In {\em STOC}, pages 275--284, 2008.

\bibitem[DM96]{DM}
John~D. Dixon and Brian Mortimer.
\newblock {\em Permutation Groups}.
\newblock Graduate Texts in Mathematics. Springer, 1996.

\bibitem[FHL80]{FurstHopcroftLuks}
Merrick Furst, John Hopcroft, and Eugene Luks.
\newblock Polynomial-time algorithms for permutation groups.
\newblock In {\em FOCS}, pages 36--41, 1980.

\bibitem[GKS06]{GKS06}
Elena Grigorescu, Swastik Kopparty, and Madhu Sudan.
\newblock Local decoding and testing for homomorphisms.
\newblock In {\em APPROX-RANDOM}, pages 375--385, 2006.

\bibitem[GL89]{GL89}
Oded Goldreich and Leonid~A. Levin.
\newblock A hard-core predicate for all one-way functions.
\newblock In {\em STOC}, pages 25--32, 1989.

\bibitem[GS14]{GS14}
Alan Guo and Madhu Sudan.
\newblock List decoding group homomorphisms between supersolvable groups.
\newblock In {\em APPROX-RANDOM}, pages 737--747, 2014.

\bibitem[Guo15]{G15}
Alan Guo.
\newblock Group homomorphisms as error correcting codes.
\newblock {\em Electronic Journal of Combinatorics}, 22(1):P1.4, 2015.

\bibitem[Kan87]{KannanIP}
Ravi Kannan.
\newblock Minkowski's convex body theorem and integer programming.
\newblock {\em Mathematics of Operations Research}, 12(3):415--440, 1987.

\bibitem[Len83]{Lenstra1983}
Hendrik~W. Lenstra, Jr.
\newblock Integer programming with a fixed number of variables.
\newblock {\em Mathematics of Operations Research}, 8:538--548, 1983.

\bibitem[Luk]{luks}
Eugene~M. Luks.
\newblock Private communication.

\bibitem[Pyb93]{PyberSubgroupsSn}
L\'aszl\'o Pyber.
\newblock Enumerating finite groups of given order.
\newblock {\em Annals of Mathematics}, 137(1):203--220, 1993.

\bibitem[Ser03]{SeressPGA}
\'Akos Seress.
\newblock {\em Permutation Group Algorithms}.
\newblock Cambridge Tracts in Mathematics. Cambridge University Press, 2003.

\end{thebibliography}
\end{document}